\documentclass[11pt]{article}

\usepackage{a4wide}
\usepackage{hyperref}
\hypersetup{colorlinks=true,citecolor=blue,linkcolor=blue,urlcolor=blue}
\usepackage{amsmath}
\usepackage{amssymb}
\usepackage{amsthm}   
\usepackage{enumerate}
\usepackage{thm-restate}
\usepackage[margin=1in]{geometry}
\usepackage{times}

\theoremstyle{plain}
\newtheorem{theorem}{Theorem}
\newtheorem{conjecture}[theorem]{Conjecture}
\newtheorem*{observation*}{Observation}
\newtheorem*{theorem*}{Theorem}
\newtheorem{corollary}[theorem]{Corollary}
\newtheorem{lemma}[theorem]{Lemma}
\newtheorem*{lemma*}{Lemma}
\newtheorem{observation}[theorem]{Observation}
\newtheorem{definition}[theorem]{Definition}
\newtheorem{remark}[theorem]{Remark}
\newtheorem{example}[theorem]{Example}

\usepackage[UKenglish]{babel}
\usepackage[UKenglish]{isodate}
\usepackage{url}
\usepackage{soul}
\usepackage[backgroundcolor=green]{todonotes}
\usepackage{graphicx}
 
\let\epsilon=\varepsilon

\def\boldx{\mathbf{x}}

\def\E{\mathbb{E}} 
\newcommand{\pr}{\mathbb{P}}

\def\calB{\mathcal{B}}

\def\calE{\mathcal{E}}
\def\calF{\mathcal{F}}

\def\calT{\mathcal{T}}

\def\Cinit{C_{\mathrm{init}}}

\def\Cblock{\kappa}

\def\Cslack{\Cblock}

\def\Einit{\calE_{\mathrm{init}}}
\newcommand{\Ejam}{\calE_{\mathrm{jam}}}

\newcommand{\revS}{\widetilde{S}}
\newcommand{\revRandom}{\widetilde{\Random}}

\newcommand{\Jmax}{J_\mathrm{max}}
 
\newcommand{\unstickProb}{p_{\mathrm{unstick}}}
\newcommand{\unstickIndicator}{\mathbf{unstick}}
\newcommand{\tauleave}{\tau_\mathrm{leave}}
 \def\tbirth{t_{\mathrm{birth}}}
 \newcommand{\tauend}{\tau_{\mathrm{end}}}
 \def\jmin{j_{\mathrm{min}}}
 \newcommand{\tauinit}{\tau_{\mathrm{init}}}
 
 \def\bins{\mathrm{bins}}
 
 \def\revT{\widetilde{\trajectory}}
 \def\revS{\widetilde{S}}
 \def\revc{\tilde{c}}
 \def\Fill{\mathrm{Fill}}
 \def\birth{\mathrm{birth}}
 \def\stuck{\mathbf{stuck}}
 \def\newstuck{\mathbf{newstuck}}
 \def\unstuck{\mathbf{unstuck}}

\newcommand{\stuckvect}{\mathbf{stuckvect}}
\newcommand{\ballvect}{\mathbf{ballvect}}
\newcommand{\balls}{\mathbf{balls}}
\newcommand{\send}{\mathbf{send}}
\newcommand{\stucksend}{\mathbf{stucksend}}
 
\newcommand{\YABfull}{T(Y^A,Y^B)}
\newcommand{\YAB}{T}
\newcommand{\Random}{R} 
\newcommand{\trajectory}{B}
\newcommand{\TrajectorySet}{\mathcal{B}}

\renewcommand{\Pr}{\pr}  

\title{Instability of backoff protocols with arbitrary arrival rates }
\author{Leslie Ann Goldberg and John Lapinskas\thanks{A short version, without the proofs, will appear in the Proceedings of SODA 2023. For the purpose of Open Access, the
authors have applied a CC BY public copyright licence to any Author Accepted Manuscript version arising
from this submission. All data is provided in full in the results section of this paper.} 
}
\date{16 February 2025}
\begin{document}
\maketitle

\begin{abstract}
In contention resolution, multiple processors are trying to coordinate to send discrete messages through a shared channel with limited communication. If two processors send at the same time, the messages collide and are not transmitted successfully. Queue-free backoff protocols are an important special case --- for example, Google Drive and AWS instruct their users to implement binary exponential backoff to handle busy periods. It is a long-standing conjecture of Aldous (IEEE Trans. Inf. Theory 1987) that no stable backoff protocols exist for any positive arrival rate of processors. This foundational question remains open; instability is only known in general when the arrival rate of processors is at least $0.42$ (Goldberg et al.\ SICOMP 2004). We prove Aldous' conjecture for all backoff protocols outside of a tightly-constrained special case using a new domination technique to get around the main difficulty, which is the strong dependencies between messages.
 \end{abstract}
 
\section{Introduction}\label{sec:intro}

In the field of contention resolution, multiple processors (sometimes called ``stations'') are trying to coordinate to send discrete messages (sometimes called ``packets'') through a shared channel called a multiple access channel. The multiple access channel is not centrally controlled and the processors cannot communicate, except by listening to the channel. The operation of the channel is straightforward. In each (discrete) time step one or more processors might send messages to the channel. If exactly one message is sent then it is delivered successfully and the sender is notified of the success. If multiple messages are sent then they collide and are not transmitted successfully (so they will have to be re-sent later). 

We typically view the entire process as a discrete-time Markov chain. At each time step, new messages arrive at processors with rates governed by probability distributions with total rate $\lambda > 0$. After the arrivals, each processor independently chooses whether to send a message through the channel. 
A \emph{contention-resolution protocol} is a randomised algorithm that the processors use to decide
when to send messages to the channel (and when to wait because the channel is too busy!). 
Our objective is to find a \emph{stable} protocol~\cite{GGMM},
which is a protocol with the property that the
corresponding Markov chain is  positive recurrent, implying that there is a stationary distribution bounding the likely extent to which messages build up over time. Other objectives include bounding the expected  
waiting time of messages and maximising the throughput, which is the rate at which messages succeed.
 
Issues of contention resolution naturally arise when designing networking protocols~\cite{wifi-IEEE,ALOHA}, but it is also relevant to hardware design that enables concurrency~\cite{hardware-lockfree,hardware-elision} and to interaction with cloud computing services~\cite{cloud-AWS, cloud-Google}. 
The example of cloud computing will be an instructive one, so we expand on it. Suppose that an unknown number of users is submitting requests to a server which is struggling under the load, with more users arriving over time. The users do not have access to load information from the server and they do not have knowledge of each other --- all they know is whether or not their own requests to the server are getting through. The central question of contention resolution is then: how often should users re-send their requests in order to get everyone's requests through as quickly as possible?

There are two main categories of contention resolution protocol --- these differ according to the extent to which processors listen to the channel. In \textit{full-sensing protocols}, processors constantly listen to the shared channel obtaining
partial information.
For example, in addition to learning
whether its own sends are successful, a
processor may learn 
on which steps the channel is quiet (with no sends)  \cite{SST} 
or it may learn on which steps
there are successful sends 
\cite{BKKP-adversarial}
or it may learn both \cite{MH}.
While full-sensing protocols are suitable in many settings, there are important settings where they cannot reasonably be implemented (such as the cloud computing example above). In \textit{acknowledgement-based protocols}, the only information that processors receive about the shared channel is whether their own messages get through.

We distinguish between two ways of modelling message arrival. The earliest work in the field focused on  \textit{queueing models}, in which the number $N$ of processors is fixed and each processor maintains a queue of messages to be sent. These models are appropriate for a static network.
A particularly simple example is the slotted ALOHA protocol~\cite{Roberts-ALOHA},  
one of the first networking protocols. In this protocol, if there are $N$ processors with non-empty queues, then
these processors   send independently with probability $1/N$. For large $N$, this is stable if $\lambda < 1/e$. However, it requires the processors to know the value of $N$. 
In order to get around this difficulty, Metcalfe and Boggs proposed \textit{binary exponential backoff}, in which a processor which has already 
had $i$ unsuccessful attempts to send a given message
waits a random amount of time (a geometric random variable with mean $2^i$) before again attempting to send.  Binary exponential backoff (with some modifications) forms the basis for Ethernet~\cite{MB-BEB} and TCP/IP~\cite{TCP}.
For any $N$, binary exponential backoff is known to be stable for sufficiently small $\lambda$  \cite{GGMM,AGM}.  
Unfortunately, this value of
$\lambda$ depends on $N$ and binary exponential backoff
is unstable if $\lambda$ is sufficiently large~\cite{HLR}.  Remarkably, H{\aa}stad, Leighton and Rogoff~\cite{HLR} showed that 
\emph{polynomial backoff} (where the waiting time after the $i$'th collision is a geometric
random variable with expectation $i^{\alpha}$ for  some $\alpha>1$) is  stable for all $\lambda \in (0,1)$. 
For contention resolution with queues
even more powerful full-sensing protocols are known   --- in particular, there is a stable full-sensing protocol even for the more general model in which some specified pairs of processors are allowed to use the channel simultaneously~\cite{ST,SST}.

In this paper, we focus on \textit{queue-free models}, which allow for dynamic networks and are more appropriate for public wi-fi hotspots~\cite{wifi-IEEE} or cloud computing~\cite{cloud-AWS, cloud-Google}. 
We again consider these models in discrete time.
In these models, processors arrive in the system according to a Poisson distribution with rate $\lambda$, and each processor only wants to send a single message rather than maintaining a queue; in fact, we typically identify the processors with the messages that they are trying to send. As usual, only one message can pass through the channel at any given time step. In this setting, an acknowledgement-based protocol can be viewed as a joint distribution $(T_1,T_2,\dots)$ of times. For each message, the corresponding processor independently samples $(\tau_1,\tau_2,\dots)$ from $(T_1,T_2,\dots)$.
If the message does not get through during the first $j-1$ times that it is sent then the processor waits for $\tau_j$ time steps before sending it for the $j$'th time. An important special case is that of \textit{backoff protocols}, in which $(T_1,T_2,\dots)$ is a tuple of independent geometric variables. Equivalently, a backoff protocol is associated with a \textit{send sequence} of probabilities $\mathbf{p}=p_0,p_1,\dots$ such that, if a processor 
has already had $j$ unsuccessful sends,   then it will send its message on the following time step with probability $p_j$; thus $1/p_j$ is the expected waiting time $\E(T_j)$. For example, the case $p_j = 2^{-j}$ gives rise to the binary exponential backoff protocol that
we have already described. This protocol  is widely used in the queue-free model: both AWS and Google 
Drive advise users to implement binary exponential backoff when using their services~\cite{cloud-AWS, cloud-Google}.

In the queue-free setting, there has been a great deal of interesting work developing full-sensing protocols and proving that these perform
well. Some of this is described in the survey
of Chlebus~\cite{Chlebus}. See also \cite{MH, EH, GFL}.
More recently,
Bender et al.~\cite{BKKP-adversarial}
have shown that
in the full-sensing model without collision detection (where processors listen to the channel to learn on which steps there are successful sends but are unable to distinguish between silence and collisions) 
there is a full-sensing protocol which achieves constant throughput 
even when the message arrival is adversarial rather than random. Chen et al.~\cite{CJZ-jamming} 
demonstrate a full-sensing protocol that can
achieve a decent throughput, even in the presence
of adversarial jamming; Bender et al.~\cite{bender2024fully} have strengthened this result to work even when each individual sender can only listen to the channel rarely. Despite these advances regarding full-sensing protocols,
and other protocols assuming more capabilities from processors than acknowledgement-based protocols
\cite{RU,GMPS}, for acknowledgement-based protocols the most fundamental possible question remains open: {\bf do stable protocols exist at all?} Indeed, this problem remains open even for backoff protocols, and 
most work on the question has focused on this case.

The following foundational conjecture was made by Aldous~\cite{Aldous} in 1987, and is widely believed. It is the focus of this work.

\begin{conjecture}[Aldous's Conjecture]\label{conj:backoff} 
    In the queue-free setting, no backoff protocol is stable for \textit{any} positive value of $\lambda$.
\end{conjecture}

Aldous's conjecture remains open to this day. It has been proved for arrival rates $\lambda \ge 0.42$~\cite{GJKP}, 
but for arbitrary arrival rates the only known results concern special cases which avoid a central difficulty inherent to the problem. Consider a backoff protocol with send sequence $p_0,p_1,\dots$ and arrival rate $\lambda > 0$. For all integers $j,t \ge 1$, write $b_j(t)$ for the set of messages in the system at time $t$ which have already sent $j$ times  (all unsuccessfully). Write $S_j(t)$ for the number of messages in $b_j(t)$ which send at time $t$,
$S_0(t)$ for the number of ``newborn'' messages which send for the first time at time~$t$,  
and $S(t) = S_0(t) + S_1(t) + \dots$ for the total number of sends from all messages in the system. Thus a message escapes the system at time $t$ if and only if $S(t) = 1$, Aldous's conjecture implies that $S(t)=1$ for less than a $\lambda$ proportion of times,  i.e.\ that messages arrive faster than they escape. 
Very often, the reason that this occurs is that $S(t) \ge 2$ for most values of $t$. However, it is not hard to show that for all $j\geq 1$ and most times $t$ we have $\E(S_j(t)) = p_j\E(|b_j(t)|) \lessapprox  \lambda$, and so on most time steps $S_j(t) = 0$; thus to show $S(t) \ge 2$ on most time steps, we must engage with the complicated joint distribution $(|b_1(t)|, |b_2(t)|,\dots)$. This is the key difficulty that 
all current arguments have avoided, which restricts the classes
of send sequences to which they apply.
The tool that enabled us to prove Aldous's conjecture for most protocols
is a new domination technique for bounding this joint distribution.
Before stating our result and the new technique we first summarise
progress that can be made without engaging with the   joint distribution.
In the following summary, we classify protocols in terms of the key quantity~$1/p_j$, which is the expected waiting time before a message sends after having its $j$'th collision. 
\begin{itemize}
\item Kelly and MacPhee~\cite{kelly-macphee} categorised the class of backoff protocols for which $S(t) \ge 2$ for all sufficiently large $t$. This result covers all protocols with subexponential expected waiting times, i.e.\ whose send sequences satisfy 
$1/p_j = o(c^j)$ as $j\to\infty$ for all $c>1$ (see Corollary~\ref{cor:backoff-kmp}). Since these protocols
are unstable in such a very strong way, Kelly and MacPhee are  able to avoid working with the joint distribution in favour of applying the Borel-Cantelli lemmas. 
\item Aldous~\cite{Aldous} proved that binary exponential backoff is unstable for all $\lambda > 0$, and his argument easily extends to all backoff protocols with exponential expected waiting times, i.e.\ whose send sequences satisfy $1/p_j = \Theta(c^{j})$ as $j\to\infty$ for some $c>1$ (see Section~\ref{sec:intro-context}). This proof relies on proving concentration for specific variables $|b_j(t)|$ as $t\to\infty$, and then applying union bounds over suitable ranges of $j$ and $t$, again avoiding the joint distribution. This concentration fails in general; for example, if $p_j \ge 3\lambda$, then $b_j(t) = \emptyset$ for most values of $t$. (See Definition~\ref{def:reliable} for a more detailed discussion.)
\item A simple argument known to the authors of~\cite{GJKP} (but not previously published) shows that there is no stable backoff protocol which has infinitely many super-exponential expected waiting times, i.e.\ with an infinite subsequence $p_{j_1}, p_{j_2}, \dots$ satisfying $1/p_{j_k} = \omega(c^{j_k})$ as $k\to\infty$ for all $c>1$. We state and prove this as Lemma~\ref{lem:killer}. 
In this case there is no need to engage with the joint distribution of the
$S_j(t)$ variables because the proof relies on bounding $S(t) \ge S_0(t)$, i.e.\ only considering newborn messages.
\end{itemize}

Unfortunately, the above results cannot be combined 
in any simple way to prove Aldous's conjecture.
For example, by 
including some $p_j$~values such that the expected waiting time
$1/p_j$ is less than exponential,
 it is easy to construct  protocols which 
neither exhibit concentration for specific variables $|b_j(t)|$ nor satisfy $S(t) \ge 2$ for all sufficiently large $t$ (see Section~\ref{sec:intro-context}), so to show that these protocols are unstable we must engage with the joint distribution.

Our main technical contribution is a proof (see Lemma~\ref{lem:simple-newdombinoms}) that, roughly speaking, we can dominate the joint distribution of $(|b_1(t)|,|b_2(t)|,\dots)$ below by a much simpler collection of independent Poisson variables whenever $\E(S(t))\to\infty$ as $t\to\infty$. Using this, we are able to almost entirely solve the problem of inconsistent decay rates and prove Aldous's conjecture except in some extreme cases characterised in  Definition~\ref{def:LCED}. 
Before describing these extreme cases, we give some easier-to-state consequences of our main result (Theorem~\ref{thm:LCED}).
In the following theorems (and throughout the paper) a \emph{backoff process} is a backoff protocol in the queue-free model.
The first consequence is that all protocols with monotonically non-increasing send sequences are unstable.

\begin{restatable}{theorem}{thmmono}
 \label{thm:mono}
For every   $\lambda \in (0,1)$ and
every   monotonically non-increasing send sequence $\mathbf{p}= p_0,p_1,\ldots$,
the backoff process
with arrival rate $\lambda$ and send sequence~$\mathbf{p}$
is  unstable.
\end{restatable}

We have included Theorem~\ref{thm:mono} because it has a clean statement, but
our proof technique doesn't rely on any kind of monotonicity.
For example, our main Theorem, Theorem~\ref{thm:LCED} also has the following corollary. 

\begin{restatable}{theorem}{thmsecond}
\label{thm:second} 
Let $\mathbf{p}$ be a send sequence.
Let $m_{\mathbf{p}}(n)$ be the median of $p_0,\ldots,p_n$.
Suppose that $m_{\mathbf{p}}(n)=o(1)$.
Then for every $\lambda\in(0,1)$
the backoff process with arrival rate $\lambda$ and
send sequence $\mathbf{p}$ is unstable. 
\end{restatable}

Of course, there is nothing very special about the median. The same would be true of sequences for which any centile is $o(1)$. At this point, we are ready to describe the extreme cases that elude our new proof technique, and
to state our main result, which shows instability except in these extreme cases. The extreme cases have the property that  the send sequence is almost entirely constant, with occasional exponential (but not super-exponential) waiting times thrown in.

\begin{restatable}{definition}{defLCED}
\label{def:LCED}
A send sequence $\mathbf{p}$ is \emph{LCED} 
(``largely constant with exponential decay'')
if it satisfies the following properties:
\begin{enumerate}[(i)]\setlength{\itemsep}{1pt}\setlength{\parskip}{0pt}\setlength{\parsep}{0pt}

\item \label{LCED:largelyconstant}
{\bf ``Largely constant'':}
For all $\eta > 0$, there exists $c>0$ such that for infinitely many $n$,
$        |\{j \le n \colon p_j > c\}| \geq (1-\eta)n$.
 
\item \label{LCED:exponentialdecay}
{\bf ``with exponential decay'':}
$\mathbf{p}$ has an infinite subsequence $(p_{\ell_1},p_{\ell_2},\dots)$ which 
satisfies $\log (1/p_{\ell_x}) = \Theta(\ell_x)$ as $x\to\infty$.
                
\item \label{LCED:notsuperexponential}
{\bf ``(but without super-exponential decay)'':}
$\log (1/p_j) = O(j)$ as $j\to\infty$.
\end{enumerate}
    \end{restatable}

As an illustrative example of item (i) taking $\eta = 999/1000$, it implies that as you progress along the send sequence $p_1,p_2,\dots$, infinitely often, you will notice that all but $1-\eta = 0.1\%$ of the $p_j$'s you have seen so far are bounded below by some constant $c$. The same holds for values of $\eta$ that are closer to 1, but $c$ will be correspondingly smaller.
Obviously, this is also true for values of~$\eta$ that are closer to~$1$ but the corresponding constant $c$ would be smaller.
Item (ii) means that there is an infinite subsequence of $j$'s where the expected waiting times (after $j$ failures) is exponentially long. Item (iii) just means that expected waiting times are not more than exponentially long.

With this definition, we can state our main theorem, which proves Aldous's  conjecture for all   sequences except 
LCED sequences
and extends all previously-known results. 

\begin{restatable}{theorem}{thmLCED}
\label{thm:LCED}
Let $\mathbf{p}$ be a send sequence which is not LCED.
Then for every $\lambda\in (0,1)$ the backoff process with arrival
rate $\lambda$ and send sequence $\mathbf{p}$ is   unstable.
\end{restatable}

As we discuss in Section~\ref{sec:intro-future}, LCED sequences can exhibit qualitatively different behaviour from non-LCED sequences, with arbitrarily long ``quiet patches'' during which almost every message that sends is successful. Lemma~\ref{lem:simple-newdombinoms}, our domination of the joint distribution $(|b_1(t)|, |b_2(t)|, \dots)$ by independent Poisson variables, is actually false during these ``quiet patches'', and so proving Aldous's conjecture   for LCED sequences will require new ideas.  On a conceptual level, the ``quiet patches'' exhibited by some LCED sequences are essentially the only remaining obstacle to a full proof of  Aldous's conjecture.

Nevertheless, as we show in this paper, our new domination is sufficient 
to cover all send sequences except the LCED sequences. Thus Lemma~\ref{lem:simple-newdombinoms} makes substantial progress on the long-standing conjecture, the first progress in many years, resulting in strong new instability results such as Theorems~\ref{thm:mono} and~\ref{thm:second}.

The remainder of the introduction is structured as follows. 
Section~\ref{sec:formal-defs} gives the formal definition of
a backoff process. Sections~\ref{sec:explainKM} and~\ref{sec:theproofsketch} sketch the proof of Theorem~\ref{thm:LCED}, with Section~\ref{sec:explainKM} giving an overview of the relevant existing proof techniques and Section~\ref{sec:theproofsketch} explaining our novel ideas. Section~\ref{sec:intro-future} 
discusses the remaining obstacles to proving Conjecture~\ref{conj:backoff}. 

\subsection{Formal definitions}\label{sec:formal-defs}

We say that a stochastic process is \textit{stable} if it is positive recurrent, and \textit{unstable} otherwise (i.e.\ if it is null recurrent or transient). A \emph{backoff process} is a backoff protocol
in the queue-free model. 

Informally, a backoff process is a discrete-time Markov chain associated with an arrival rate $\lambda \in (0,1)$ and a send sequence $\mathbf{p} = p_0,p_1,p_2,\ldots$ of
real numbers in the range $(0,1]$. Following Aldous~\cite{Aldous}, we identify processors and messages, and we think of these as balls moving through a sequence of bins. Each time a message sends, if no other message sends at the same time step, it leaves the system; otherwise, it moves to the next bin. Thus at time $t$, the $j$'th bin contains all messages which have sent $j$ times without getting through (these sends occurred at time steps up to and including time $t$). The system then evolves as follows at a time step $t$. First, new messages are added to bin 0 according to a Poisson distribution with rate $\lambda$. Second, for all $j \ge 0$, each message in bin $j$ sends independently with probability $p_j$. Third, if exactly one message sends then it leaves the system, and otherwise all messages that sent from any bin~$j$
move to the next bin, bin~$j+1$.

\begin{remark}
There is no need to consider arrival rates $\lambda \geq 1$ because 
it is already known that backoff processes with arrival rate $\lambda\geq 1$
are unstable~\cite{GJKP}. We also don't allow $p_j=0$ since that would trivially cause transience (hence, instability).
\end{remark}

\noindent
\textbf{Formal definition of backoff processes.}
A \emph{backoff process} with \emph{arrival rate}~$\lambda \in (0,1)$ and \emph{send sequence}~$\mathbf{p}=p_0,p_1,p_2,\dots \in (0,1]$ is
a stochastic process $X$
defined as follows.
 \emph{Time steps}~$t$ are positive integers.
 \emph{Bins} $j$ are non-negative   integers. 
There is an infinite set of 
 \emph{balls}. 
We now define the set $b_j^X(t)$, which will be the set of balls in bin~$j$ just after (all parts of) the $t$'th step.
  Initially, all bins are empty, so for all  non-negative integers~$j$,
$b_j^X(0)=\emptyset$.
  For any positive integer~$t$, the $t$'th step of~$X$ involves
(i) step initialisation (including birth), 
(ii) sending, and (iii) adjusting the bins.   
Step~$t$ proceeds as follows.
  \begin{itemize}

\item 
Part (i) of step $t$ (step initialisation, including birth):
An integer $n_t$ is chosen independently from a Poisson distribution
with mean~$\lambda$. This is the number of \emph{newborns} at time~$t$.
The set $b^{\prime X}_0(t)$ contains the balls in $b^X_0(t-1)$ together with
$n_t$ new balls which are \textit{born} at time~$t$. 
For each $j\geq 1$ we define $b^{\prime X}_j(t)  = b_j^X(t-1)$.

\item Part (ii) of step $t$ (sending): For all $j\geq 0$, all balls in $b^{\prime X}_j(t)$ 
send independently with probability $p_j$. 
We use $\send^X(t)$ for the set of balls that send at time~$t$.

\item Part (iii) of step $t$ (adjusting the bins): 
\begin{itemize}
\item  If $|\send^X(t)|\leq1$ then 
any ball in $\send^X(t)$ 
\emph{escapes} 
so for all $j\geq 0$ we define
$b_j^X(t) =  
b^{\prime X}_j(t) \setminus \send^X(t)$.

\item Otherwise, no balls 
escape but balls that send
move to the next bin, so we define
$b_0^X(t) = b_0^{\prime X}(t)\setminus \send^X(t)$ and, 
for all $j\geq 1$,
$b_j^X(t) = (b_{j-1}^{\prime X}(t) \cap \send^X(t) ) \cup
(b_j^{\prime X}(t) \setminus \send^X(t))$.
 \end{itemize}

\end{itemize}

Finally, we define $\balls^X(t) = \cup_j b_j^X(t)$.

\subsection{Technical context}\label{sec:intro-context}\label{sec:explainKM}
 
We first formally state the result alluded to in Section~\ref{sec:intro} which proves instability for backoff protocols whose send sequences decay super-exponentially.
\begin{restatable}{lemma}{lemkiller}\label{lem:killer} 
 Let $X$ be a backoff process with arrival rate~$\lambda\in (0,1)$ and send sequence 
 $\mathbf{p}=p_0,p_1,\ldots$. 
 If, for infinitely many~$j$, 
 $p_j \le (\lambda p_0/2)^j$,   then $X$ is  unstable.
 \end{restatable}
We defer the proof to Section~\ref{sec:prelim}, but it is simple. Essentially, we dominate the expected time for a newborn ball to leave the process below under the assumption that a sending ball always leaves the process unless another ball is born at the same time. Under the assumptions of Lemma~\ref{lem:killer}, this is infinite. 

We next describe the results of Kelly and MacPhee~\cite{kelly-macphee} and Aldous~\cite{Aldous} in more detail than the previous section. This will allow us in Section~\ref{intro:technical-contributions} to clearly identify the regimes in which instability is not known, and to clearly highlight the novel parts of our arguments.

We first introduce a key notion from Aldous~\cite{Aldous}. Informally, an \emph{externally-jammed process} is a backoff process in which balls never leave; thus if  a single ball sends at a given time step, it moves to the next bin as normal. (See Section~\ref{sec:externally-jammed} for a formal definition.) Unlike backoff processes, an externally-jammed process starts in its stationary distribution; thus for all $j \ge 0$, 
the size of $b_j(0)$ is drawn from a Poisson distribution with mean $\lambda/p_j$. There is a natural coupling between a backoff process $X$ and an externally-jammed process $Y$ such that $|b_j^X(t)| \le |b_j^Y(t)|$ for all $j$ and $t$ (see Observation~\ref{obs:ext-jammed-useful}); thus an externally-jammed process can be used to dominate (from above) the number of balls in a backoff process.

As discussed earlier, Kelly and MacPhee~\cite{kelly-macphee} gives a necessary and sufficient condition for infinitely many messages to get through. In our context, the relevant case of their result can be stated as follows. Given a send sequence $\mathbf{p}$, let $W_0,W_1,\dots$ be independent geometric variables such that $W_j$ has parameter $p_j$ for all $j$. Then for all $\tau \ge 0$, we define
\[
    \mu_\tau(\mathbf{p}) = \sum_{j=0}^\infty \pr\Big(\sum_{k=0}^j W_k \le \tau\Big).
\]
Thus if a ball is born at time $t$ in an externally-jammed process,
$\mu_\tau(\mathbf{p})$ is the expected number of times that ball sends up to time $t+\tau$.

\begin{theorem}[{\cite[Theorem~3.10]{kelly-macphee}}]\label{thm:full-kmp}
    Let $\mathbf{p}$ be a send sequence, and suppose that for all $\lambda \in (0,1)$,
    \[
        \sum_{\tau=0}^\infty \mu_\tau(\mathbf{p}) e^{-\lambda \mu_\tau(\mathbf{p})} < \infty.
    \]
Then for all $\lambda\in(0,1)$, the backoff process $X$ with 
arrival rate~$\lambda$ and send sequence $\mathbf{p}$ is unstable. Moreover, with probability $1$, only finitely many balls leave $X$.
\end{theorem}

The following corollary is proved in Section~\ref{sec:additional}.

\begin{restatable}{corollary}{backoffkmp}\label{cor:backoff-kmp}
Let $\mathbf{p}$ 
be a send sequence such that 
$\log(1/p_j) = o(j)$ as $j\to\infty$. 
Then for all $\lambda\in(0,1)$, the backoff process $X$ with
arrival rate~$\lambda$ and send sequence~$\mathbf{p}$
is unstable.
\end{restatable}

The   
result that is proved
in Aldous~\cite{Aldous} is the instability of binary exponential backoff, in which $p_j = 2^{-j}$ for all $j$, for all arrival rates $\lambda > 0$.
Aldous's paper says  ``without checking the details, [he] believe[s] the argument could be modified to show instability for [all backoff protocols]''.
Unfortunately, this turns out not to be accurate. However, the proof does generalise in a natural way to cover a broader class of backoff protocols. Our own result extends this much further 
building on new ideas, but it is  nevertheless based on a similar underlying framework, so we now discuss Aldous's proof in more detail.

A key notion is that of \textit{noise}. Informally, given a state $\mathbf{x}(t)$ of a backoff process at time $t$, the noise 
(Definition~\ref{def:noise})
of $\mathbf{x}(t)$ is given by
\[
    f(\mathbf{x}(t)) = \lambda p_0 + \sum_{j=0}^\infty p_j|b_j^X(t)|.
\] 

Thus, the noise of a backoff process at time $t$ is the expected number of sends at time $t+1$ conditioned on the state at time~$t$.
Unsurprisingly, while the noise is large, multiple balls are very likely to send at each time step  so balls leave the process very slowly (see Lemma~\ref{lem:one_onesided}).

A slightly generalised version of Aldous's proof works as follows. The first step is to choose $j_0$ to be a suitably large integer, and wait until a time $t_0$ at which bins $1,\dots,j_0$ are all ``full'' in the sense that $|b_j(t_0)| \ge c\lambda/p_j$ for 
some constant~$c$ and
all $j \le j_0$. (Such a time $t_0$ exists with probability 1. Observe that in the stationary distribution of an externally-jammed process, bin $j$ contains $\lambda/p_j$ balls in expectation.)
We then define  $\tau_1,\tau_2,\dots$ with $\tau_\ell = C\sum_{j=0}^{j_0+\ell} (1/p_j)$ for some constant~$C$; thus $\tau_\ell$ is $C$ times the expected number of time steps required for a newborn ball to reach bin $b_{j_0+\ell}$ in a jammed channel. The key step of the proof is then to  argue that with suitably low failure probability, for all $\ell$ and all $t$ satisfying $t_0+\tau_\ell \le t \le t_0 + \tau_{\ell+1}$, all bins $j$ with $(j_0+\ell)/10 \le j \le j_0+\ell$ satisfy $|b_j(t)| \ge \zeta\lambda/p_j$ for some constant $\zeta$. In other words, we prove that with suitably low failure probability, there is a slowly-advancing frontier of bins which are always full; in particular the noise increases to infinity, and the process is transient and hence unstable.

In order to accomplish this key step of the proof, the argument is split into two cases. If $t$ is close to $t_0$, there is a simple argument based on the idea that bins $1$ through $j_0$ were full at time $t_0$ and, since $j_0$ is large, they have not yet had time to fully empty. For the second case, suppose that $t$ is significantly larger than $t_0$, and 
the goal is (for example) to show that bin $j$ is very likely to be full at time $t = t_0+\tau_\ell$. By structuring the events carefully, one may assume that the process still has large noise up to time $t-1$, so with high probability not many balls leave the process in time steps $\{t_0+1,\dots,t\}$. At this point, Aldous uses a time-reversal argument to show that, in the externally-jammed process, bin $j$ fills with balls that are born after time~$t_0$
during an interval of $\tau_\ell$ time steps.  Under a natural coupling, these balls follow the same trajectory in both the backoff process and the externally-jammed process unless they leave the backoff process; thus, by a union bound, with high probability most of these balls are present in bin $j$ at time $t$ as required. Aldous then applies a union bound over all bins $j$ with $(j_0+\ell)/10 \le j \le j_0+\ell$; crucially, the probability bounds for each individual bin $j$ are strong enough that he does not need to engage with the 
more complicated joint distribution of the 
contents of the bins.

Since Aldous works only with binary exponential backoff in~\cite{Aldous}, he takes $\tau_\ell \sim 2^{j_0+\ell}$; his conjecture that his proof can be generalised to all backoff protocols is based on the idea that the definition of $\tau_\ell$ could be modified, which gives rise to the more general argument above. Unfortunately, there are very broad classes of backoff protocols to which this generalisation cannot apply. We now define the collection of send sequences which this modified version of Aldous's proof could plausibly handle.

\begin{definition}\label{def:reliable}
    A send sequence $\mathbf{p}$ is \emph{reliable} if it has the following property. Let $\lambda > 0$, and let $X$ be a backoff process with arrival rate $\lambda$ and send sequence $\mathbf{p}$. Then with positive probability, there exists $\zeta > 0$, and times $t_0,t_1,t_2,\dots$ and collections $\calB_0,\calB_1,\calB_2, \dots$ of bins such that:
    \begin{itemize}
        \item for all $i$, for all $t$ satisfying $t_i \le t \le t_{i+1}$, for all $j \in \calB_i$, we have $|b_j(t)| \ge \zeta \lambda/p_i$.
        \item $|\calB_i|$ is an increasing sequence with $|\calB_i| \to \infty$ as $i\to\infty$.
    \end{itemize}
\end{definition}

Indeed, if Aldous's proof works for a send sequence $\mathbf{p}$, then it demonstrates that $\mathbf{p}$ is reliable: Aldous takes $\zeta = 1/2$, $t_i = \tau_i$ for all $i$, and $\calB_i = \{(j_0+i)/10,\dots,j_0+i\}$ for all $i$. In short, a reliable protocol is one in which, after some ``startup time'' $t_0$, at all times one can point to a large collection of bins which will reliably \textit{all} be full enough to provide significant noise. 
As discussed above, it is important that they are all full --- otherwise it is necessary to delve into the complicated inter-dependencies of the bins. 

\subsection{Proof sketch}\label{sec:theproofsketch}
\label{intro:technical-contributions}

In order to describe the proof of Theorem~\ref{thm:LCED} and explain how it works on unreliable send sequences that are not covered by Aldous's proof techniques, we first set out a guiding example by defining the following family of send sequences. 

\begin{example}\label{example:guide}
Given $\rho \in (0,1)$, an increasing sequence $\mathbf{a}$ of non-negative integers with $a_0 = 0$, and a function $g\colon \mathbb{N}\to\mathbb{N}$, the \emph{$(\rho,\mathbf{a},g)$-interleaved send sequence} $\mathbf{p}=p_0,p_1,\ldots$ is given by
\[
    p_j = \begin{cases}
        \rho^j & \mbox{ if }a_{2k} \le j \le a_{2k+1}-1\mbox{ for some }k\ge 0,\\
        g(j) & \mbox{ if }a_{2k+1} \le j \le a_{2k+2}-1\mbox{ for some }k\ge 0.
    \end{cases}
\]
Thus the $(\rho,\mathbf{a},g)$-interleaved send sequence is an exponentially-decaying send sequence with base $\rho$ spliced together with a second send sequence specified by $g$, with the splices occurring at points given by $\mathbf{a}$.  

In this section, we will take $g(j) = 1/\log\log j$ and $a_k = 2^{2^k}$.
We will refer to the  $(\rho,\mathbf{a},g)$-interleaved send sequence
with this choice of~$g$ and~$\mathbf{a}$ as a
\emph{$\rho$-interleaved send sequence}. 
\end{example}

Observe that a $(\rho,\mathbf{a},g)$-interleaved send sequence  
fails to be LCED
whenever $g(j)=o(1)$.
Thus,  a $\rho$-interleaved send sequence   is not LCED.

We claim that as long as $\rho$ is sufficiently small,
a $\rho$-interleaved send sequence is not reliable and neither Theorem~\ref{thm:full-kmp} nor Lemma~\ref{lem:killer} prove instability; we expand on this claim in Section~\ref{sec:intro-future}. 
Our main technical result, which we introduce next,
will be strong enough  to show that backoff protocols with these send
sequences are unstable. In order to introduce the technical result
we first give some definitions.
 
\begin{restatable}{definition}{defCblock}\label{def:Cblock}
Given an $\eta \in (0,1)$, let $\Cblock(\eta) =   \lceil 3/\eta \rceil$.
\end{restatable}

\begin{restatable}{definition}{defpstar}\label{def:pstar}
Let 
$$p_*(\lambda,\eta,\nu) = 
\min\left\{ \frac{\lambda}{200},\, 
\frac{\lambda \eta}
{1800 \Cblock(\eta)^2 \log(1/\nu)} \right\} 
.$$
\end{restatable}

\begin{restatable}{definition}{defsuitable}\label{def:suitable}
Fix $\lambda$, $\eta$ and $\nu$ in $(0,1)$. 
A send sequence $\mathbf{p} = p_0,p_1,\ldots$
is \emph{$(\lambda,\eta,\nu)$-suitable} if 
$p_0=1$ and
there is a positive integer $n_0$
such that  for all $n\geq n_0$,
\begin{itemize}
    \item  $|\{ j \in [n] \mid p_j \leq p_*(\lambda,\eta,\nu)\}| > \eta n$,
    and
    \item  
$\nu^n < p_n$.
\end{itemize}
\end{restatable} 

\begin{theorem}\label{thm:simple-technical}
    Fix $\lambda$, $\eta$ and $\nu$ in $(0,1)$. Let $\mathbf{p}$ be a $(\lambda,\eta,\nu)$-suitable send sequence. Let $X$ be a
    backoff process with arrival rate $\lambda$ 
    and send sequence~$\mathbf{p}$.
    Then $X$ is transient, and hence unstable.
\end{theorem}

Theorem~\ref{thm:LCED} will follow from  Theorem~\ref{thm:simple-technical} (actually, from a slightly more technical version of it, stated as Corollary~\ref{cor:last}), from  Corollary~\ref{cor:backoff-kmp}, and from Lemma~\ref{lem:killer}. We give the proof of Theorem~\ref{thm:LCED}, along with the proofs of Theorems~\ref{thm:mono} and~\ref{thm:second}, in Section~\ref{sec:mainproofs}. 
Observe that for all $\lambda,\eta > 0$, 
every $\rho$-interleaved send sequence
is a $(\lambda, \eta, 2\rho)$-suitable send sequence, so Theorem~\ref{thm:simple-technical} does indeed apply.
 
Fix $\lambda > 0$. Before sketching the proof of Theorem~\ref{thm:simple-technical}, we set out one more piece of terminology. Recall that there is a natural coupling between a backoff process $X$ and an externally-jammed process $Y$ with arrival rate $\lambda$ and send sequence $\mathbf{p}$. Under this coupling, the evolution of $X$ is a deterministic function of the evolution of $Y$ --- balls follow the same trajectories in each process, except that they leave $X$ if they send at a time step in which no other ball sends. As such, we work exclusively with $Y$, calling a ball \textit{stuck} if it is present in $X$ and \textit{unstuck} if it is not. We write $\stuck_j^Y(t)$ for the set of stuck balls in bin $j$ at time $t$, and likewise $\unstuck_j^Y(t)$ for the set of unstuck balls. (See the formal definition of an externally-jammed process in Section~\ref{sec:externally-jammed}.)

Intuitively, the main obstacle in applying Aldous's proof sketch to
a $\rho$-interleaved send sequence $\mathbf{p}$ is the long sequences of bins with $p_j = 1/(\log \log j)$. In the externally-jammed process $Y$, each individual bin $j$ is likely to empty of balls --- and hence also of stuck balls --- roughly once every $(\log j)^\lambda$ steps. (Indeed, in the stationary distribution, bin $j$ is empty with probability $e^{-\lambda/ g(j)} = (\log j)^{-\lambda}$.) As such, the individual bins do not provide a consistent source of noise as reliability would require. We must instead engage with the joint distribution of bins and argue that a long sequence of bins with large $p_j$, taken together, is likely to continue providing noise for a very long time even if individual bins empty. Unfortunately, these bins are far from independent: for example, if $|\stuck_j^Y(t)|=|\stuck_{j+1}^Y(t)|=0$, then we should also expect $|\stuck_{j+2}^Y(t)|=0$. By working only with individual reliably-full bins, Aldous was able to sidestep this issue with a simple union bound, but this is not an option for a $\rho$-interleaved send sequence~$\mathbf{p}$.

Observe that since $Y$ starts in its stationary distribution, for all $t \ge 0$ the variables $\{|b_j^Y(t)|\colon j \ge 0\}$ are mutually independent Poisson variables with $\E(|b_j^Y(t)|) = \lambda/p_j$. 
Intuitively, we might hope that while the noise of $Y$ is high and very few balls are becoming unstuck, most balls in $Y$ will remain stuck from birth and we will have $|\stuck_j^Y(t)| \approx |b_j^Y(t)|$. If this is true, we can hope to dominate the variables $|\stuck_j^Y(t)|$ below by mutually independent Poisson variables as long as the noise remains large. Our largest technical contribution --- and the hardest part of the proof --- is making this idea rigorous, which we do in Lemma~\ref{lem:simple-newdombinoms}. (Strictly speaking, we use a slightly more technical version stated as Lemma~\ref{lem:newdombinoms}.) This lemma is the heart of our proof; we spend Section~\ref{sec:coupling-defs} setting out the definitions needed to state it formally, then give the actual proof in Section~\ref{sec:rev}. With Lemma~\ref{lem:newdombinoms} in hand, we then prove Theorem~\ref{thm:technical} in Section~\ref{sec:provemain}, using a version of Aldous's framework in which a union bound over large bins is replaced by Chernoff bounds over a collection of independent Poisson variables which dominate a non-independent collection of smaller bins from below.

For the rest of this section, we will focus on the statement and proof of Lemma~\ref{lem:newdombinoms}. A fundamental difficulty is that we can only hope for such a domination to work while the noise of $Y$ is large, but the noise is just a weighted sum of terms $|\stuck_j^Y(t)|$, which are precisely the variables we are concerned with --- we therefore expect a huge amount of dependency. To resolve this, we define the \textit{two-stream process} in Section~\ref{sec:two-stream}. Essentially, we split $Y$ into a pair $T=(Y^A,Y^B)$ of two externally-jammed processes $Y^A$ and $Y^B$, each with arrival rate $\lambda/2$, in the natural way. We then say that a ball becomes unstuck in $Y^A$ if it sends on a time step at which no other ball in $Y^B$ sends, and likewise for $Y^B$. 
Thus only balls from $Y^{\boldsymbol{B}}$ can prevent balls in $Y^{\boldsymbol{A}}$ from becoming unstuck, and only balls from $Y^{\boldsymbol{A}}$ can prevent balls in $Y^{\boldsymbol{B}}$ from becoming unstuck. Each stream of the process acts as a relatively independent source of collisions for the other stream while its noise is high. Naturally, there is a simple coupling back to the originally externally-jammed process under which stuck balls in $Y$ are dominated below by stuck balls in $T$.

We must also explain what we mean by $T$ having ``large noise''. We give an informal overview here, and the formal definitions in Section~\ref{sec:jammedness}. We define a ``startup event'' $\Einit(t_0)$ which occurs for some $t_0$ with probability 1, and which guarantees very large noise at time $t_0$ (governed by a constant $\Cinit$). We divide the set of bins into \textit{blocks} $B_1,B_2,\dots$ of exponentially-growing length, and define a map $\tau\mapsto \bins(\tau)$ from times $t_0+\tau$ to blocks $B_i$. We define a constant $\zeta > 0$, and say that $T$ is \emph{$t_0$-jammed for $\tau$} if both $Y^A$ and $Y^B$ have at least $\zeta |\bins(\tau-1)|$ noise at time $t_0+\tau-1$. (Observe that if every ball were stuck, then in expectation the balls in $\bins(\tau-1)$ would contribute $\lambda|\bins(\tau-1)|$ total noise, so jammedness says that on average the bins in $\bins(\tau-1)$ are ``almost full''.) We can now state a somewhat simplified version of Lemma~\ref{lem:newdombinoms}.

\begin{lemma}\label{lem:simple-newdombinoms}
 Fix $\lambda$, $\eta$ and $\nu$ in $(0,1)$  and a $(\lambda,\eta,\nu)$-suitable send sequence~$\mathbf{p}$ with $p_0=1$.
 Let $Y^A$ and $Y^B$ be independent externally-jammed processes with arrival rate $\lambda/2$ and send sequence $\mathbf{p}$  and consider the two-stream externally-jammed process $\YAB=\YAB(Y^A,Y^B)$.  Let $t_0$ and $\tau$ be sufficiently large integers. Then there is a coupling of
\begin{itemize}
    \item      $\YAB$
 conditioned on $\Einit(t_0)$, 
 \item a sample $\{Z_j^{A} \mid j \in \bins(\tau)\}$ where each $Z_j^{A}$ is
 chosen independently from a Poisson distribution with mean $\lambda/(4 p_j)$, and 
 \item a sample $\{Z_j^{B} \mid j \in \bins(\tau)\}$ where each $Z_j^{B}$ is chosen from a Poisson distribution with mean $\lambda/(4 p_j)$ and these are independent of each other but not of the $\{Z_j^A\}$ values
 \end{itemize}
  in   such a way that  
 at least one of the following happens:
 
   \begin{itemize}
     \item $T$ is not $t_0$-jammed for $\tau$, or
    \item for all $j\in \bins(\tau)$, 
    $|\stuck_j^{\YAB}(A,t_0+\tau)  | \geq Z_j^A$
    and $|\stuck_j^{\YAB}(B,t_0+\tau)  | \geq Z_j^B$.
 \end{itemize}
 \end{lemma}

To prove Lemma~\ref{lem:simple-newdombinoms}
(or its full version Lemma~\ref{lem:newdombinoms}), we first formalise the idea that $Y^A$ can act as a ``relatively independent'' source of noise to prevent balls in $Y^B$ becoming unstuck (and vice versa). In Section~\ref{sec:defRandom}, we define a \textit{random unsticking process}; this is an externally-jammed process in which balls become unstuck on sending based on independent Bernoulli variables, rather than the behaviour of other balls. In Lemma~\ref{lem:domcouplenew}, we dominate the stuck balls of $Y^A$ below by a random unsticking process $R$ for as long as $Y^B$ is ``locally jammed'', and likewise for the stuck balls of~$Y^B$.

In order to analyse the random unsticking process, we then make use of a time reversal, defining a \emph{reverse random unsticking process} $\tilde{R}$ in Section~\ref{sec:rev-def} which we couple to $R$ using a probability-preserving bijection between ball trajectories set out in Lemma~\ref{lem:timerev}. We then exploit the fact that balls move independently and the fact that the number of balls following any given trajectory is a Poisson random variable to prove Lemma~\ref{lem:simple-newdombinoms}, 
expressing the set of balls in $\stuck_j^{\tilde{R}}$  
as a sum of independent Poisson variables (one per possible trajectory) with total mean at least $\lambda/(4p_j)$. Unfortunately, the coupling between $R$ and $\tilde{R}$ does not run over all time steps --- for example, since we need to condition on $\Einit(t_0)$, it certainly cannot cover times before $t_0$ in $R$. As a result, we cannot analyse all trajectories this way; instead, we define a subset $\Fill_j^R(t)$ of balls in bin $j$ of $R$ at time~$t$ which are born suitably late in $R$, and analyse the corresponding set $\Fill_j^{\revRandom}(t)$ of balls in $\revRandom$ (see Section~\ref{sec:fill-def}). We then prove in Lemma~\ref{lem:revfill} that the variables $|\Fill_j^{\revRandom}(t)|$ for $j \in \bins(\tau-1)$ are dominated below by a tuple of independent Poisson variables; this bound then propagates back through the couplings to a lower bound on $|\Fill_j^R(t)|$ in $R$ and finally to the required bounds on $|\stuck^{Y^A}(t)|$ and $|\stuck^{Y^B}(t)|$ in $T$ given by Lemma~\ref{lem:simple-newdombinoms}. 
  
\subsection{Future work}\label{sec:intro-future}

It is natural to wonder where the remaining difficulties are in a proof of Aldous's conjecture. LCED sequences are quite a restricted case, and at first glance one might suspect they could be proved unstable by incremental improvements to Theorem~\ref{thm:simple-technical}, Corollary~\ref{cor:backoff-kmp}, and Lemma~\ref{lem:killer}. Unfortunately, this is not the case --- LCED send sequences can exhibit qualitatively different behaviour than the send sequences covered by these results.

Recall the definition of a $(\rho,\mathbf{a},g)$-interleaved send sequence from
Example~\ref{example:guide}.  Let $\mathbf{p}$ be such a send sequence, and let $X$ be a backoff process with send sequence $\mathbf{p}$ and arrival rate $\lambda \in (0,1)$. 
Unless $g$ is exponentially small, Lemma~\ref{lem:killer} does not apply to $\mathbf{p}$ for any $\lambda < 2\rho$. 
Also, Theorem~\ref{thm:full-kmp} (and hence Corollary~\ref{cor:backoff-kmp}) does not apply to $\mathbf{p}$ for any $\lambda \in (0,1)$ 
whenever $\rho$ is sufficiently small and $\mathbf{a}$ grows suitably quickly. To see this, suppose that $\rho$ is small and that $\mathbf{a}$ grows quickly.
Recall the definition of $\mu_\tau(\mathbf{p})$ from Theorem~\ref{thm:full-kmp}. It is not hard to show (e.g.{} via union bounds) 
that $\mu_\tau(\mathbf{p}) \le \log_{1/\rho}(\tau) + O(1)$ for all $\tau$ satisfying $a_{2k} \le \log_{1/\rho}(\tau)\le a_{2k+1}/3$ for any $k$; 
from this, it follows that
\[
    \sum_{\tau=0}^\infty \mu_\tau(\mathbf{p}) e^{-\lambda \mu_\tau(\mathbf{p})} \ge \sum_{\tau=0}^\infty \mu_\tau(\mathbf{p}) e^{-\mu_\tau(\mathbf{p})} = \infty \mbox{ for all $\lambda \in (0,1)$.}
\]
Since Theorem~\ref{thm:full-kmp} is both a necessary and a sufficient condition for a backoff process to send only finitely many times (see~\cite{kelly-macphee}), we can conclude that the methods of Kelly and MacPhee do not apply here. 

As discussed in Section~\ref{sec:theproofsketch}, $\mathbf{p}$ is not LCED whenever $g(j) = o(1)$ as $j\to\infty$. It is also easy to see that $\mathbf{p}$ is LCED whenever $g(j) = \Theta(1)$ and $a_{2k+1}-a_{2k} = o(a_{2k+2}-a_{2k+1})$ as $k\to\infty$, so in this case our own Theorem~\ref{thm:LCED} does not apply. To show that a simple generalisation of Theorem~\ref{thm:simple-technical} will not suffice, we introduce the following definition.

\begin{definition}
    A send sequence $\textbf{p}$ \emph{has quiet periods} if it has the following property for some $\lambda \in (0,1)$. Let $X$ be a backoff process with arrival rate $\lambda$ and send sequence $\mathbf{p}$. Then with probability 1, there exist infinitely many time steps on which $X$ has noise less than $1$.
\end{definition}

Recall that the proof of Theorem~\ref{thm:simple-technical} (as with Aldous's result~\cite{Aldous}) relies on proving that the noise of a backoff process increases to infinity over time, so it is unsuitable for
send sequences with quiet periods. 
Moreover, Corollary~\ref{cor:backoff-kmp} is based on Kelly and MacPhee's necessary and sufficient condition for a backoff protocol to have only finitely
many successful sends (with probability~$1$). Backoff protocols with quiet periods have infinitely many successful sends, so cannot be proved unstable
by the methods of Corollary~\ref{cor:backoff-kmp}. Finally, Lemma~\ref{lem:killer} handles many   protocols with quiet periods (e.g.\ $p_j = 2^{-2^j}$), but 
these protocols  are \textit{always} quiet, and 
have the property that they can be proved unstable by simple domination
arguments (assuming that
a ball always succeeds when it sends 
unless another ball is born at the same time). For many send sequences with quiet periods, these assumptions are not true.

The key problem is that if $g(j) = 1/2$ (say), $\rho$ is sufficiently small, and $\mathbf{a}$ grows sufficiently quickly, then $\mathbf{p}$ is very likely to have quiet periods. 
Indeed, in the externally-jammed process, for large values of $k$ all bins $1,\dots,a_{2k}-1$ are likely to simultaneously empty (which happens with probability at least roughly $2^{-a_{2k}}$ in the stationary distribution) long before bin $a_{2k}$ fills with stuck balls and begins contributing significant noise (which will take time at least $(1/\rho)^{a_{2k}}$). Heuristically, we would expect $\mathbf{p}$ to alternate between ``quiet periods'' with very low noise and ``noisy periods'' in which noise first grows rapidly, then stays high for a very long time. Our current methods are incapable of handling this juxtaposition --- in particular, since messages can easily escape from the system in quiet periods, any argument based around our ideas would need to bound the frequency and duration of quiet periods.

However, despite this, we believe we have made significant progress towards a proof of Aldous's conjecture even for LCED sequences. Recall from Section~\ref{sec:theproofsketch} that the key part of our proof of Theorem~\ref{thm:simple-technical} is Lemma~\ref{lem:newdombinoms}, which gives a way of dominating the number of stuck balls in different bins below by an \textit{independent} tuple of Poisson variables as long as overall noise is large. Importantly, the proof of Lemma~\ref{lem:newdombinoms} doesn't actually rely on the fact that $\mathbf{p}$ is suitable. Suitability is important to determining the parameters of the block construction, which is vital in applying Lemma~\ref{lem:newdombinoms} to prove Theorem~\ref{thm:simple-technical}, but the Lemma~\ref{lem:newdombinoms} proof goes through regardless of what these parameters are. As such, we believe a slightly more general version of Lemma~\ref{lem:newdombinoms} will be important even for LCED sequences as a tool for analysing ``noisy periods''. The main remaining challenge is in showing that such ``noisy periods'' are common enough to outweigh the effect of ``quiet periods'' of very low noise.

The remainder of the paper is structured as follows.
In Section~\ref{sec:prelim}, we set out some standard preliminaries and prove Lemma~\ref{lem:killer} and Corollary~\ref{cor:backoff-kmp}. In Section~\ref{sec:coupling-defs} we introduce the main random processes and block construction we will need in order to state Lemma~\ref{lem:newdombinoms}, the key lemma for our proof of Theorem~\ref{thm:simple-technical} (see Section~\ref{sec:theproofsketch}). 
In Section~\ref{sec:rev}, we state and prove Lemma~\ref{lem:newdombinoms}. In Section~\ref{sec:provemain}, we use Lemma~\ref{lem:newdombinoms} to prove Theorem~\ref{thm:technical}. Finally, in Section~\ref{sec:mainproofs}, we use Theorem~\ref{thm:technical} (together with Lemma~\ref{lem:killer} and Corollary~\ref{cor:backoff-kmp}) to prove Theorem~\ref{thm:LCED}, and derive Theorems~\ref{thm:mono} and~\ref{thm:second} as corollaries.

\section{Preliminaries}\label{sec:prelim}\label{sec:KMP-killer}
\label{sec:KM}\label{sec:additional}

For all positive integers $n$, we define $[n]$ to be the set $\{1,2,\dots,n\}$.
All logarithms are to the base~$e$ unless specified otherwise.

A stochastic process is ``stable'' if it is positive recurrent. It is ``unstable'' otherwise. An unstable process can be either null recurrent or transient,
We now define the filtration $\calF^X_0,\calF^X_1,\ldots$
of a backoff process~$X$. 

\begin{definition}\label{def:backoff-filt}
    Let $X$ be a backoff process.
    Since all bins are initially empty, $\calF^X_0$ contains no information.
    For each positive integer~$t$, 
    $\calF^X_t$ contains $\calF^X_{t-1}$ 
    together with ${b'}_0^X(t)$ 
    and, for each $j\geq 0$, $b_j^X(t)$.
    It is clear that $\send^X(t)$ and $\balls^X(t)$ can be deduced from $\calF^X_t$.
\end{definition}

We now state some common Chernoff bounds.

\begin{lemma}\label{lem:chernoff-small-dev} (E.g.\ \cite[Theorems~4.5 and 5.4]{MUbook})
Let $\Psi$ be a real random variable with mean $\mu$ which is either Poisson or a sum of independent Bernoulli variables. Then for all $0 < \delta < 1$, we have  $\pr(\Psi \le (1-\delta)\mu) \le e^{-\delta^2\mu/2}$.
\end{lemma}

\begin{lemma}\label{lem:chernoff-large-upper}
    Let $\Psi$ be a binomial random variable with mean $\mu$. Then for all $x > 1$, we have
    \[
        \pr(\Psi \ge x\mu) < e^{{-}\mu x(\log x-1)}
    \]
\end{lemma}
\begin{proof}
    By the Chernoff bound of~\cite[Theorem~4.4(1)]{MUbook}, for all $\delta>0$, we have
    \[
        \pr\big(\Psi \ge (1+\delta)\mu\big) < \Big(\frac{e^\delta}{(1+\delta)^{1+\delta}}\Big)^\mu.
    \]
    Taking $x = 1+\delta$ and putting everything under the exponent implies that for all $x > 1$,
    \[
        \pr(\Psi \ge x\mu) < e^{\mu(x-1) - \mu x\log x} < e^{-\mu x(\log x-1)},
    \]
    as required.
\end{proof}

We will also use some preliminary lemmas. We first prove the following lemma known to to the authors of~\cite{GJKP} and stated in Section~\ref{sec:intro-context}.

\lemkiller*

\begin{proof}
Suppose 
that for infinitely many~$j$ we have
$p_j \le (\lambda p_0/2)^j$. 
We will show that $X$ is not positive recurrent by showing that the expected time for $X$ to return to its initial state (the state containing no balls) is infinite.

Let $\calT_0$ be the first time such that $b_1^X(\calT_0) \ne \emptyset$, and note that 
$\calT_0 < \infty$ with probability 1. Consider any positive integer~$t_0$ and any value $F_{t_0}^X$  of $\calF_{t_0}^X$ such that $\calT_0=t_0$. We will now consider the process~$X$ conditioned on $\calF_{t_0}^X = F_{t_0}^X$. Let 
\[
    \calT = \min(\{\infty\} \cup \{t > t_0 \mid \balls^X(t) = \emptyset\}).
\]
Then by the definition of positive recurrence, it suffices to show that $\E[\calT \mid \calF_{t_0}^X = F_{t_0}^X] = \infty$.

Let $\beta$ be an arbitrary ball in $b_1^X(t_0)$, and let 
$$\calT' = \min(\{\infty\} \cup \{t>t_0 \mid \mbox{$\beta \in \send^X(t)$ 
and no newborns send at~$t$} \}).$$
Observe that $\beta$ must escape in order for $X$ to return to its initial state, and $\beta$ cannot escape before time $\calT'$, so $\calT' \le \calT$; it therefore suffices to show that $\E[\calT' \mid \calF_{t_0}^X = F_{t_0}^X] = \infty$.

At any given step $t> t_0$,
the number of newborns that send is a Poisson random variable
with parameter~$\lambda p_0$, 
so the probability of at least one newborn sending on any given step is $r=1-e^{-\lambda p_0}$. We therefore have
$$\E[\calT'\mid \calF_{t_0}^X = F_{t_0}^X] = \sum_{j\geq 1} r^{j-1} (\tfrac{1}{p_1} + \cdots + 
\tfrac{1}{p_j} ) (1-r) =  \frac{1-r}{r} \sum_{j\geq 1} \frac{r^j}{p_j}.$$
By the assumption in the lemma statement,
infinitely many values of $j$ contribute a summand
which is at least $(2r/(\lambda p_0))^j$ which is at least~$1$
since 
$r = 1-e^{-\lambda p_0} \geq \lambda p_0/2$, thus the sum diverges as required.
\end{proof}

Before setting up the structure for our own proof, we also
   prove Corollary~\ref{cor:backoff-kmp}, a Corollary
of Kelly and MacPhee's Theorem~\ref{thm:full-kmp}.

\backoffkmp*
\begin{proof}
    Let $\mathbf{p}$ be a send sequence as in the corollary statement. As in Theorem~\ref{thm:full-kmp}, let $W_0,W_1,\dots$ be independent geometric variables such that $W_j$ has parameter $p_j$ for all $j$. For all $\tau \ge 2/p_0$, let 
    \begin{equation}\label{eq:mutau-def}
        \mu_\tau(\mathbf{p}) = \sum_{j=0}^\infty \Pr\Big(\sum_{k=0}^j W_k \le \tau\Big),
    \end{equation}
    and let
    \[
        M(\tau) = \max\Big\{j \ge 0 \colon \sum_{k=0}^j (1/p_k) \le \tau/2\Big\}.
    \]
    
    In order to apply Theorem~\ref{thm:full-kmp}, we first prove that for all $\tau \ge 2/p_0$, $\mu_\tau(\mathbf{p}) \ge M(\tau)/2$. By the definition of $M(\tau)$, for all $j \le M(\tau)$ we have $\E(\sum_{k=0}^j W_k) \le \tau/2$. It follows by Markov's inequality that for all $j \le M(\tau)$,
    \[
        \pr\Big(\sum_{k=0}^j W_k \le \tau\Big) = 1 - \pr\Big(\sum_{k=0}^j W_k > \tau\Big) \ge 1/2.
    \]
    It now follows from~\eqref{eq:mutau-def} that
    \begin{equation}\label{eq:mutau-bound-1}
        \mu_\tau(\mathbf{p}) \ge \sum_{j=0}^{M(\tau)} \Pr\Big(\sum_{k=0}^j W_k \le \tau\Big) \ge M(\tau)/2.
    \end{equation}
    
    We now prove that $M(\tau) = \omega(\log\tau)$ as $\tau\to\infty$. Let $C>1$ be arbitrary; we will show that when $\tau$ is sufficiently large, we have
    \[
        \sum_{k=0}^{\lceil C\log \tau\rceil}(1/p_k) \le \tau/2,
    \]
    and hence that $M(\tau) \ge C\log\tau$ as required. Recall that $\log(1/p_j)=o(j)$ as $j\to\infty$, and let $j_C$ be such that $1/p_j \le e^{j/C}/(Ce^4)$ for all $j \ge j_C$. Suppose $\tau \ge 2/p_0$ is large enough that
    \[
        \sum_{k=0}^{j_C-1} (1/p_k) \le \tau/4.
    \]
    We then have
    \begin{align*}
        \sum_{k=0}^{\lceil C\log \tau\rceil}(1/p_k) 
        \le \frac{\tau}{4} + \sum_{k=j_C}^{\lceil C\log \tau\rceil}(1/p_k) 
        \le \frac{\tau}{4} +  \frac{1}{Ce^4}\sum_{k=0}^{\lceil C\log\tau\rceil} e^{k/C}
        < \frac{\tau}{4} + \frac{1}{Ce^4}\cdot\frac{e^{(1+\lceil C\log\tau\rceil)/C}}{e^{1/C}-1}.
    \end{align*}
    Since $e^x \ge 1+x$ for all $x \ge 0$, 
     $C(e^{1/C}-1) \geq 1$ and since $C>1$
      it follows that
    \[
        \sum_{k=0}^{\lceil C\log \tau\rceil}(1/p_k) \le \frac{\tau}{4} + \frac{1}{e^4}\cdot e^{2+\log\tau} < \frac{\tau}{2}.
    \]
    Hence $M(\tau) \ge C\log\tau$, so $M(\tau) = \omega(\log\tau)$ as $\tau\to\infty$, as claimed; by~\eqref{eq:mutau-bound-1}, it follows that $\mu_\tau(\mathbf{p}) = \omega(\log \tau)$ as $\tau\to\infty$.

Finally, we observe that, since $\mu_\tau(\mathbf{p}) = \omega(\log \tau)$, 
for all $\lambda \in (0,1)$
there is a $\tau_\lambda$ such that, for $\tau \geq \tau_\lambda$, 
$\mu_\tau(\mathbf{p}) \geq (2/\lambda) \log \tau$. Thus
\[
\sum_{\tau=0}^\infty \mu_\tau(\mathbf{p})e^{-\lambda\mu_\tau(\mathbf{p})} \le \Theta(1) + \sum_{\tau=\tau_\lambda}^\infty \frac{(2/\lambda) \log \tau} {\tau^2} < \infty.
\]
The result therefore follows immediately from Theorem~\ref{thm:full-kmp}.
\end{proof}

 \section{Preparation for our Technical Result}\label{sec:coupling-defs}

We will prove 
Theorem~\ref{thm:simple-technical}  
by studying stochastic processes that are related to, but
not exactly the same as, backoff processes.
To do this, it will be useful to assume that the
send sequence $\mathbf{p}$ has $p_0=1$.
The next observation 
sets up the machinery that  enables  this.

\begin{observation}\label{obs:p0one}
Fix $\lambda \in (0,1)$
Let $\mathbf{p}$ be a send sequence.
Let $\mathbf{p}'$ be the send sequence that is identical to $\mathbf{p}$
except that $\mathbf{p}'_0=1$.
Let $X$ be a backoff process with arrival rate~$\lambda$ and send sequence~$\mathbf{p}$.
Let $X'$ be a backoff process with arrival rate~$\lambda p_0$ and send
sequence~$\mathbf{p}'$.
Then there is a coupling of 
$X$ and $X'$ such that, for 
every positive integer $t$ and every positive integer $j$,
$b_j^X(t) \supseteq b_j^{X'}(t)$.
\end{observation}
  
Indeed, in defining the coupling, we can identify the newborns of $X'$ with the newborns of $X$ which immediately send.

 \subsection{Definition of externally-jammed process}\label{sec:externally-jammed}

In order to study backoff processes
it will be helpful to define another process
that is associated with  an arrival rate $\lambda \in (0,1)$
and a send sequence $\mathbf{p} = p_0,p_1,p_2,\ldots$
of real numbers in the range $(0,1]$ 
with $p_0=1$. 

Note that in a backoff process~$X$ with $p_0=1$
every bin~$b_0^X(t)$ is empty, and this is why
bins are positive integers~$j$ in the
following definition.

$Y$ is an \emph{externally-jammed process} 
with \emph{arrival rate}~$\lambda\in (0,1)$
and \emph{send sequence} $\mathbf{p}$
if it behaves as follows.  \emph{Time steps}~$t$ are positive integers.
  \emph{Bins} $j$ are positive integers. We reserve $j$ as an index for bins, and use $\ell$ when we need a second index.
  Initialisation: For every  positive integer~$j$,
an integer $x_{j}$ is chosen independently from a
Poisson distribution with mean $\lambda/p_j$. The set
$b_j^Y(0)$ 
is the set containing $x_j$ balls which are ``born in process~$Y$ at time~$0$ in bin~$j$''.
Formally, the names of these balls are 
$ (Y,0,j,1),\dots,(Y,0,j,x_j)$.
   For any positive integer~$t$, the $t$'th step of~$Y$ involves
(i) step initialisation (including birth), 
(ii) sending, and (iii) adjusting the bins.  
The set $b_j^Y(t)$ is the  set of balls
in bin~$j$ just after (all parts of) the $t$'th step.
Step~$t$ proceeds as follows.
 
 \begin{itemize}

\item 
Part (i) of step $t$ (step initialisation, including birth):
An integer $n_t$ is chosen independently from a Poisson distribution
with mean~$\lambda$.
The set $b^{\prime Y}_0(t)$ 
of \emph{newborns} at time $t$
is the set containing $n_t$ balls which are ``born in $Y$ at time~$t$''.
Formally, the names of these balls are   
$(Y,t,1),\dots,(Y,t,n_t)$.
Then for all $j\geq 1$, $b^{\prime Y}_j(t)  = b_j^Y(t-1)$ .

\item Part (ii) of step $t$ (sending): For all $j\geq 0$, all balls in $b^{\prime Y}_j(t)$ 
send independently with probability $p_j$. 
We use $\send^Y(t)$ for the set of balls that send at time~$t$.

\item Part (iii) of step $t$ (adjusting the bins): 
For all $j\geq 1$,
$$b_j^Y(t) = (b^{\prime Y}_{j-1}(t) \cap \send^Y(t) ) \cup
(b^{\prime Y}_j(t) \setminus \send^Y(t)).$$

\end{itemize}

Note that 
 for any distinct $j$ and $j'$, $b_j^Y(t)$ is disjoint from
 $b_{j'}^Y(t)$. Also, since $p_0=1$, 
$\cup_{j\geq 1} b_j^Y(t) = \cup_{j\geq 1} b^{\prime Y}_j(t)$.
We now define the filtration $\calF^Y_0,\calF^Y_1,\ldots$
for the externally-jammed process~$Y$.
$\calF^Y_0$ contains, for each $j\geq 1$, $b_j^Y(0)$.
Then for each positive integer~$t$, 
$\calF^Y_t$ contains $\calF^Y_{t-1}$ 
together with $b^{\prime Y}_0(t)$ 
and, for each $j\geq 1$, $b_j^Y(t)$.
It is clear that $\send^Y(t)$ can be deduced from $\calF^Y_t$.
We   define  
 $\balls^Y(t) = \cup_j b_j^Y(t)$.
This can also be deduced from $\calF^Y_t$. 
We will use the following notation: Each ball $\beta = (Y,0,j,x)$
 has $\birth^Y(\beta)=0$.
 Each ball $\beta = (Y,t,x)$ has $\birth^Y(\beta) = t$.

It is going to be convenient in our analysis to think about
the balls of an externally-jammed process~$Y$ as having two
states, ``\emph{stuck}'' and ``\emph{unstuck}''. The set
$\stuck^Y(t)$ will be the set of all
balls in $\balls^Y(t)$  
that are stuck 
in the process~$Y$ at time~$t$
and $\unstuck^Y(t) = \balls^Y(t) \setminus \stuck^Y(t)$. 
(Given this identity, we will define either $\unstuck^Y(t)$
or $\stuck^Y(t)$ for each~$t$, the other can be deduced.)
We will assign the states ``stuck'' and ``unstuck'' in such a way that 
$\stuck^Y(t)$  
and $\unstuck^Y(t)$
can be deduced from $\calF^Y_t$, so the
assignment of states to balls is just a convenience.

First, 
$\stuck^Y(0) = \emptyset$.
Then for any $t\geq 1$ we have the following definition.
\begin{itemize}
 \item ${\newstuck}^Y(t) = \stuck^Y(t-1) \cup b^{\prime Y}_0(t)$.
 \item $\stucksend^Y(t) = {\newstuck}^Y(t) \cap \send^Y(t)$.
 \item If $|\stucksend^Y(t)| = 1$ then 
$\unstuck^Y(t) = \unstuck^Y(t-1) \cup  \stucksend^Y(t)$.
Otherwise, 
$\unstuck^Y(t) = {\unstuck}^Y(t-1)$.
\end{itemize}

For $t\geq 1$ we define $\stuck_j^Y(t) = \stuck^Y(t) \cap b_j^Y(t)$.
 Note that the stuck balls in an externally-jammed
 process correspond to a backoff process. This is captured in the following observation.

\begin{observation}\label{obs:ext-jammed-useful}
For any backoff process $X$ with arrival rate $\lambda \in(0,1)$ and send sequence $\mathbf{p} = p_0,p_1,\dots$ with $p_0=1$, there is a coupling of~$X$ with an externally-jammed process $Y$ with arrival rate $\lambda$ and send sequence $\mathbf{p}$ such that for all $j \ge 1$ and $t \ge 0$, $|b^X_j(t)| = |\stuck_j^Y(t)|$.
\end{observation}

 \subsection{A two-stream externally-jammed process}\label{sec:two-stream}
   
Observation~\ref{obs:ext-jammed-useful} shows
that in order to study backoff processes we can
instead study externally-jammed processes and focus on the stuck balls.

Typically, it is difficult to study externally-jammed processes due to correlation between balls. To alleviate this, we define
another related process, called a two-stream externally-jammed process. The idea is that the 
balls are (arbitrarily) divided into two streams.
Keeping track of separate streams enables domination by
processes with more independence.
   
Here is the definition.
We build a  \emph{two-stream externally-jammed process
$\YAB=\YABfull$ } with \emph{arrival rate}~$\lambda$
and \emph{send sequence} $\mathbf{p}$ 
by combining two independent externally-jammed processes $Y^A$ and $Y^B$, each  with arrival rate $\lambda/2$ and send sequence $\mathbf{p}$.
The balls of $\YAB$ are the balls of $Y^A$ together with the balls of $Y^B$. We refer to the former as \emph{$A$-balls} of $\YAB$ and to the latter
as \emph{$B$-balls} of $\YAB$.
For all $t\geq 0$ and $j\geq 1$, $b_j^{\YAB}(t) = b_j^{Y^A}(t) \cup b_j^{Y^B}(t)$.
For all $t\geq 1$ and $j\geq 0$, $b^{\prime \YAB}_j(t) = b^{\prime Y^A}_j(t) \cup b^{\prime Y^B}_j(t)$.
For all $t\geq 1$, $\send^{\YAB}(t) = \send^{Y^A}(t) \cup \send^{Y^B}(t)$.

We define
$\balls^{\YAB}(t) = \cup_j b_j^{\YAB}(t)$.
For $C\in \{A,B\}$,
$\balls^{\YAB}(C,t)$ is the set of all $C$-balls in
 $\balls^{\YAB}(t)$. Also, each
  $C$-ball $\beta$ of $\YAB$ has $\birth^{\YAB}(\beta) = \birth^{Y^C}(\beta)$.

Lemma~\ref{lem:coupleY} below
contains the following easy observation:
Let $Y^A$ and $Y^B$ be independent externally-jammed processes with arrival rate~$\lambda/2$ and consider the two-stream externally-jammed process $\YAB = \YABfull$. 
Let $Y$ be an externally-jammed process with arrival rate $\lambda$.
There is a coupling of $Y$ with $T$
which
provides, for every non-negative integer~$t$,
  a bijection~$\pi_t$ from
$\balls^{Y}(t)$ to $\balls^T(t)$
such that   $\beta \in b_j^{Y}(t)$ iff $\pi_t(\beta) \in b_j^T(t)$.
Given this, it may seem unclear what the point is of defining
the two-stream externally jammed process, since this is essentially ``the same''
as an externally jammed process.
The only difference is the names of the balls!
The point of it is that we will assign the
states stuck and unstuck to balls of~$\YAB$ in a different, and more useful, way.

For $C\in \{A,B\}$, we will now define the set
$\stuck^{\YAB}(C,t)$, which will be the set of all
 balls in $\balls^{\YAB}(C,t)$  
that are \emph{stuck} 
in the process~$\YAB$ at time~$t$.
We also define $\stuck^{\YAB}(t) = \stuck^{\YAB}(A,t) \cup \stuck^{\YAB}(B,t)$.
Similarly to externally jammed processes, it will always be the case that
$\unstuck^{\YAB}(C,t) = \balls^{\YAB}(C,t) \setminus 
  \stuck^{\YAB}(C,t)$,
so it is only necessary to define
  $\stuck^{\YAB}(C,t)$ or $\unstuck^{\YAB}(C,t)$
  and the other is defined implicitly.

For any $C\in \{A,B\}$ we will define 
$\stuck^{\YAB}(C,0) =     \emptyset$.
Then for any $t\geq 1$   we have the following definitions:

\begin{itemize}
 
\item  For $C\in \{A,B\}$, ${\newstuck}^{\YAB}(C,t) = \stuck^{\YAB}(C,t-1) \cup b^{\prime Y^C}_0(t) $.

\item For $C\in \{A,B\}$, $\stucksend^{\YAB}(C,t) = {\newstuck}^{\YAB}(C,t) \cap \send^{\YAB}(t)$.

\item Case 1:
If, for distinct $C$ and $C'$ in $\{A,B\}$,
$\stucksend^{\YAB}(C,t) = \emptyset$ 
and 
$\stucksend^{\YAB}(C',t) \neq \emptyset$,
then 
\begin{itemize}
\item $\beta$ is chosen uniformly at random from
$\stucksend^{\YAB}(C',t) $
and 
\item $\unstuck^{\YAB}(C',t) = {\unstuck}^{\YAB}(C',t-1)\cup \{\beta\}$
and 
\item $\unstuck^{\YAB}(C,t) = {\unstuck}^{\YAB}(C,t-1)$.
\end{itemize}

\item Case 2:  Otherwise, for each $C\in \{A,B\}$
$\unstuck^{\YAB}(C,t) = {\unstuck}^{\YAB}(C,t-1)$.

\end{itemize}

Although the sets $b_j^{\YAB}(t)$, $b^{\prime \YAB}_j(t)$,
$\send^{\YAB}(t)$, and $\balls^{\YAB}(C,t)$  
can be deduced from  $\calF^{Y^A}_t$ and $\calF^{Y^B}_t$,
the sets $\stuck^{\YAB}(C,t)$  
cannot. This is because  $\calF^{Y^A}_t$ and $\calF^{Y^B}_t$
do not capture any information about the choice
of the random ball~$\beta$ 
(from $\stucksend^{\YAB}(A,t)$ or $\stucksend^{\YAB}(B,t)$)
that may become unstuck at time~$t$.
Thus, we take   $\calF^{\YAB}_t$ to be
$\calF^{Y^A}_t$, $\calF^{Y^B}_t$
and the sets $\stuck^{\YAB}(A,t')$ and $\stuck^{\YAB}(B,t')$ 
for $t'\leq t$
so that all information about the evolution
of the process~$\YAB$, including the assignment of states
to the balls, is captured.

 For $C\in \{A,B\}$, $t\geq 1$ and $j \ge 1$, define
$\stuck_j^{\YAB}(C,t) = \stuck^{\YAB}(C,t)\cap b_j^{\YAB}(t)$.

We can now couple an externally-jammed process~$Y$ to
a two-stream externally-jammed process $\YAB$ in the following lemma.

\begin{lemma}\label{lem:coupleY}
Let $Y^A$ and $Y^B$ be independent externally-jammed processes with arrival rate~$\lambda/2$ and consider the two-stream externally-jammed process $\YAB = \YABfull$. 
Let $Y$ be an externally-jammed process with arrival rate $\lambda$.
There is a coupling of $Y$ with $T$
which
provides, for every non-negative integer~$t$,
  a bijection~$\pi_t$ from
$\balls^{Y}(t)$ to $\balls^T(t)$
such that the following two properties hold.
\begin{itemize}
\item $\beta \in b_j^{Y}(t)$ iff $\pi_t(\beta) \in b_j^T(t)$, and
\item $\beta \in \unstuck^{Y}(t)$ implies $\pi_t(\beta) \in \unstuck^T(t)$.
\end{itemize}
\end{lemma} 

\begin{proof}
We first define the bijections~$\pi_t$. In fact, it is easier to define the inverses~$\pi_t^{-1}$.
 
In the initialisation of $\YAB$, 
for each $C\in \{A,B\}$, let $x_j^C$ denote the number of
$C$-balls born in $\YAB$ at time~$0$ in bin~$j$.
Then $x_j$, the number of balls born in~$Y$ at time $0$ in bin $j$,
is given by $x_j = x_j^A + x_j^B$.
The bijection $\pi^{-1}_0$ 
from $\balls^T(0)$ to $\balls^Y(0)$
then maps the ball $(Y^A,0,j,x)$ to $(Y,0,j,x)$
and the ball $(Y^B,0,j,x)$ to $(Y,0,j,x_j^A+x)$
(for each $x$).

For any $t\geq 1$, the bijection $\pi^{-1}_t$ 
from $\balls^T(t)$ to $\balls^Y(t)$
is the same as $\pi_{t-1}$
except that it maps the balls in 
${b'_0}^{Y^A}(t) \cup {b'_0}^{Y^B}(t)$ 
to ${b'_0}^Y(t)$
using a similar re-naming of balls.

Note that both properties hold at time $t=0$
since all balls are unstuck at time~$0$.

Suppose, for a positive integer~$t$,
that the first $t-1$ steps of $Y$ and $T$ have
been coupled so that the two properties hold. 
We will show how to couple the $t$'th step of the process
so that the two properties again hold.
 
The first property allows us to couple the send decisions
so that
$\beta \in \send^Y(t)$ iff $\pi_t(\beta) \in \send^T(t)$.
This allows us to conclude that the first property holds at time~$t$.

To finish, we wish to show that any ball $\beta \in \unstuck^Y(t)$ has $\pi_t(\beta) \in \unstuck^T(t)$. If $\beta \in \unstuck^Y(t-1)$ then this is immediate by the second property at time $t-1$; we may therefore focus on the case where $\beta$ only becomes unstuck at time $t$, i.e., where $\beta \in \stucksend^Y(t) \cap \unstuck^Y(t)$. By the definition of the $Y$ process,
$\beta\in \stucksend^{Y}(t) \cap \unstuck^{Y}(t)$
implies that $\stucksend^{Y}(t) = \{\beta\}$.

By the first property at time~$t$ and the second property at time~$t-1$,
if any ball $\beta^*$ satisfies
$\beta^* \in \stucksend^T(t) $ then
$\pi_t^{-1}(\beta^*) \in \stucksend^Y(t)$
so $\pi_t^{-1}(\stucksend^T(t) ) \subseteq \stucksend^Y(t)$.
Hence $|\stucksend^T(t)| \le 1$. 
Now we are finished because $\pi_t(\beta)$ is either in 
$\unstuck^T(t-1)$ (in which case it is also in $\unstuck^T(t)$)
or it is in $\stucksend^T(t)$ (in which case 
$\stucksend^T(t)=  \{\pi_t(\beta)\}$
so $\pi_t(\beta) \in \unstuck^T(t)$).

\end{proof}

\subsection{Ball vectors and noise}

So far, we have defined backoff processes, 
externally-jammed processes, and
two-stream externally-jammed processes.

Here we collect some notation about these processes that we haven't needed so far, but which will be useful later.

For an externally-jammed process~$Y$
we   define
\begin{align*}
\ballvect^Y(t) &= (b_1^Y(t), b_2^Y(t), \dots),\\
\stuckvect^Y(t) &= (\stuck_1^Y(t), \stuck^Y_2(t), \dots).
\end{align*}
Similarly, for a two-stream externally-jammed process~$\YAB=\YABfull$
we   define 
\begin{align*}
    \ballvect^{\YAB}(t) &= (b_1^{\YAB}(t), b_2^{\YAB}(t), \dots),\\
    \stuckvect^{\YAB}(C,t) &= 
   (\stuck_1^{\YAB}(C,t), 
   \stuck_2^{\YAB}(C,t), \dots) \mbox{ for each }C \in \{A,B\}.
\end{align*}
We define $\stuckvect^{\YAB}(t)$ as the position-wise sum of 
$\stuckvect^{\YAB}(A,t)$ and $\stuckvect^{\YAB}(B,t)$.

\begin{definition}\label{def:noise}
Suppose that $\boldx=(x_1,x_2,\ldots)$ is a sequence of non-negative integers.
The \emph{noise} of $\boldx$
is defined by $f(\boldx) = \lambda + \sum_{i\geq 1} x_i p_i$.
\end{definition}
Note that for an externally-jammed process~$Y$ with arrival rate~$\lambda$,
$f(\stuckvect^Y(t-1))$ is 
the expected size of 
$  \stucksend^Y(t) $ conditioned on the filtration $\calF^Y_{t-1}$. Similarly, for a two-stream externally-jammed process~$\YAB$ with arrival rate~$\lambda$,
$f(\stuckvect^{\YAB}(t-1))$ is 
the expected size of 
$  \stucksend^{\YAB}(A,t) \cup 
\stucksend^{\YAB}(B,t)$ conditioned on the filtration $\calF^{\YAB}_{t-1}$.

\subsection{Statement of our technical result} \label{sec:technicalresult}
 
Recall the definition of $(\lambda,\eta,\nu)$-suitable
(Definition~\ref{def:suitable}).
 The following Theorem and   Corollary form the centre 
 of our argument. 
 Corollary~\ref{cor:last}  immediately implies Theorem~\ref{thm:simple-technical}, the main technical theorem
 from the introduction.
 Theorem~\ref{thm:technical}
  and Corollary~\ref{cor:last} are proved in Section~\ref{sec:mainproofs}.

 \begin{restatable}{theorem}{thmtechnical}
\label{thm:technical}
Fix  $\lambda$, $\eta$ and $\nu$ in $(0,1)$.
Let $\mathbf{p}$ be a $(\lambda,\eta,\nu)$-suitable send sequence.
Let $T$ be a two-stream externally-jammed process with arrival rate $\lambda$ and send sequence $\mathbf{p}$.
Let $\mathcal{R}^T = \{t \mid \stuck^T(t) = \emptyset \}$.
With probability~$1$, $\mathcal{R}^T$ is finite.
\end{restatable}

   \begin{restatable}{corollary}{corlast}\label{cor:last}
Fix $\lambda$, $\eta$ and $\nu$ in $(0,1)$.
Let $\mathbf{p}$ be a $(\lambda,\eta,\nu)$-suitable send sequence.
Let $X$ be a
backoff process with arrival rate $\lambda$ 
and send sequence~$\mathbf{p}$.
Let $\mathcal{R}^{X} = \{t \mid \balls^{X}(t)=\emptyset\}$.
Then, with probability $1$, $\mathcal{R}^{X}$ is finite.
Hence, $X$ is transient so it is unstable.
 \end{restatable}

\subsection{\texorpdfstring{Definitions depending on a $(\lambda,\eta,\nu)$-suitable send sequence $\mathbf{p}$ }{Definitions depending on a suitable send sequence p}}\label{sec:jammedness}\label{sec:implicitconsts}
  
 We will need lots of definitions to prove  Theorem~\ref{thm:technical}.
 For everything leading up to and including
 the proof of Theorem~\ref{thm:technical}, fix 
 an arrival rate $\lambda\in(0,1)$ 
 a real number $\eta \in (0,1)$,
 a real number $\nu\in (0,1)$,
 and
 a $(\lambda,\eta,\nu)$-suitable send sequence~$\mathbf{p}$.
 Since $\lambda$, $\eta$ and $\nu$ are
  fixed, we will use $\Cblock$ as shorthand for $\Cblock(\eta)$ and we will use $p_*$ as
  a shorthand for $p_*(\lambda,\eta,\nu)$.
  All of the following notation will depend on 
 $\mathbf{p}$, $\lambda$, $\eta$, and~$\nu$.

 \begin{definition}\label{def:j0} 
 $j_0$ is a constant such that for every $j\geq j_0$ we have $p_j >  \nu^j$.
 \end{definition}
 \begin{definition}\label{def:jmin}
  $\jmin = \min\{j \mid p_j < 1\}$.  
  \end{definition}
We know that $j_0$ and $\jmin$ exist since $\mathbf{p}$ is $(\lambda,\eta,\nu)$-suitable.

Before proceeding with the rest of our formal definitions, we give a brief sense of what their purpose will be. Our proof will proceed by partitioning the bins of $T$ into contiguous ``blocks'' $B_1,B_2,\dots$ of exponentially growing length, with block $B_i$ containing roughly $\Cblock^i$ bins. We will wait until, by chance, at some time $t_0$, some large initial segment $B_1,\dots,B_{I_0}$ of bins contains many stuck $A$-balls and $B$-balls; we will denote the number of balls required by $\Cinit$. Since we will choose $\Cinit$ to be large, these stuck balls are likely to take a long time to move out of $B_1,\dots,B_{I_0}$, say until time $t_0 + \tauinit$. Until they do so, the $A$-balls will interfere with any balls escaping from $Y^B$, and crucially they will do so independently from the evolution of $Y^B$; likewise, the $B$-balls will interfere with any balls escaping from $Y^A$. 

We will use this interference to argue that $B_{I_0+1}$ is likely to fill with newborn stuck $A$-balls and $B$-balls by some time $t_0+\tau_1 < \tau+\tauinit$, that $B_{I_0+2}$ is likely to fill with stuck $A$-balls and $B$-balls by some time $t_0+\tau_2 < \tau+\tauinit$, and so on. After time $t+\tauinit$, we will show that the process becomes self-sustaining, and that $B_i$ being full of stuck $A$-balls and $B$-balls is enough to fill $B_{i+1}$ by itself. We will formalise this idea at the end of the section with the definition of ``jammedness''.

We now formalise this block construction and define our constants.

\begin{definition}  \label{def:block}
  For each  positive integer~$i$, the \emph{block}
    $B_i$ is a contiguous interval of positive integers (called bins). This block  
has bins $\ell(i),\ldots,u(i)$ which are defined as follows.
First, $u(0)=0$ and
$\ell(1)=u(1)=1$. Then for $i>1$, $\ell(i)=u(i-1)+1$ and
$u(i) = \Cblock u(i-1)$.
 \end{definition}

 \begin{definition}\label{def:weight}
  The \emph{weight} of a bin $j$ is $w_j = 1/p_j$.
The \emph{weight} of a block $B_i$ is $W_i = \sum_{j \in B_i} 1/p_j$.
     \end{definition}

Observe that the weight of a block $B_i$ is the expected number of time steps required  for a ball to pass through it, from $b_{\ell(i)}$ to $b_{u(i)+1}$.

 \begin{definition}\label{def:zeta}
  $\zeta= \eta \lambda/24$.
  \end{definition}
The constant $\zeta$ will control how full we require the bins of a given block to be in order to consider $T$ ``jammed''.

\begin{definition}\label{def:I0}
Let $I_0$ be the smallest integer such that
\begin{itemize}
\item $I_0 \geq \jmin$.
\item  For all  $n\geq \ell(I_0)$:
\begin{itemize}
    \item  $|\{ j \in [n] \mid p_j \leq p_*(\lambda,\eta,\nu)\}| > \eta n$,
    and
    \item  
$\nu^n < p_n$.
\end{itemize}
\item   $\zeta |B_{I_0}| \geq 4$.
\item  
For $c = \zeta \frac{\Cblock-1}{16\Cblock^{{2}}}$,
we have $I_0 \geq \log(4/c)/c$ 
and $\exp(I_0 c)\geq 4I_0$.
\end{itemize}
\end{definition}  

Definition~\ref{def:suitable} ensures that the
second condition can be satisfied.
The final condition is satisfied if $I_0$ is sufficiently large. It ensures that, 
for all $i \ge I_0$, 
$  4 i \leq 
   \exp({ \zeta 
   i \frac{\Cblock-1}{16\Cblock^{{2}}}})
   $. 
It is not actually important that $I_0$ is the smallest
integer satisfying the properties in Definition~\ref{def:I0},
it is merely important that all of these properties
(which are monotonically increasing in~$I_0$) are satisfied.

\begin{definition}\label{def:taui}
 Define $\tau_0=0$.  
For every positive integer~$i$, define
  $\tau_i = 
        \Cslack\sum_{k=1}^{I_0+i}(I_0+i-k+1)\lceil W_k \rceil$.
\end{definition}

 \begin{definition}\label{def:Ioftau}
 Define $I(0) = I_0+1$.
 For every positive integer~$\tau$,
 $I(\tau)\geq I_0+1$ 
 is the integer such that 
 $$\tau_{I(\tau)-I_0-1} \le \tau < \tau_{I(\tau)-I_0}.$$
\end{definition} 
 
\begin{definition}\label{def:bins}
For every non-negative integer~$\tau$,
$\bins(\tau) = B_{I(\tau)-1}$. 
\end{definition}

Intuitively, referring back to the above sketch, $I(\tau)$ is chosen so that at time $t_0+\tau$, we expect $B_{I(\tau)-1} = \bins(\tau)$ to be full of stuck balls and $B_{I(\tau)}$ to be in the process of filling. (Recall that $W_i$ is the expected time required for a ball to travel through block $B_i$, and observe that $\tau_i - \tau_{i-1} = \Cblock\sum_{k=1}^{I_0+i} \lceil W_k\rceil$.)
  
 \begin{definition}
  \label{def:tauinit}
 Let $\tauinit$ be the smallest integer such that
  
\begin{itemize}
\item  $\tauinit \geq \max\{
10^7/\lambda^2, 20,  (2\Cblock/(1-\nu))^4
\}$.   

\item $I(\tauinit) \geq \max\{I_0+3,
2 I_0 (2 Q + 1)
\}$ where  $Q = \max_{k=1}^{I_0} \lceil W_k \rceil$. 
\end{itemize}
\end{definition}

Once again, it is not important that $\tauinit$ is the smallest integer satisfying the properties in 
Definition~\ref{def:tauinit}, it is just important that these properties are satisfied, and they can all be satisfied by making~$\tauinit$ sufficiently large, 
since $I(\tau)$ is $\omega(1)$ as a function of~$\tau$.

\begin{definition}
\label{def:Cinit} 
Let  
  \begin{equation*} 
        \Cinit = \bigg\lceil \max\bigg( 
        \frac{\zeta |B_{I_0}| }{ p_{\jmin}},
        \frac{12\log(100\tauinit)}{(1-p_{\jmin})^{\tauinit}}, \frac{2\zeta|\bins(\tauinit)|}{p_{\jmin}(1-p_{\jmin})^{\tauinit}} 
         \bigg)\bigg\rceil.
    \end{equation*}
\end{definition}    
    
    Recall that $p_{\jmin} < 1$ by the definition of $\jmin$, so 
the definition of~$\Cinit$ in Definition~\ref{def:Cinit}  is well-defined.
 
The following definitions refer
to a  
two-stream externally-jammed
process $ \YAB=\YABfull$   with arrival rate~$\lambda$ and send sequence~$\mathbf{p}$.

\begin{definition}\label{def:Einit}
 For all integers~$t \ge 1$,
$\Einit^{\YAB}(t)$ 
is the event that  for all $C\in \{A,B\}$,
$|\stuck_{\jmin}^{\YAB}(C,t)  | \geq \Cinit$.
  \end{definition} 

  \begin{definition}\label{def:jammed}
For $C\in \{A,B\}$ and  positive integers~$t$ and~$\tau$,
the  
two-stream 
externally-jammed process $ \YAB$ is 
\emph{$(C,t)$-jammed for~$\tau$}  if
$f(\stuckvect^{ \YAB}(C,t+\tau-1))\geq \zeta |\bins(\tau-1)|$. 
It is \emph{$t$-jammed for~$\tau$}   if it is $(A,t)$-jammed and $(B,t)$-jammed at~$\tau$.
\end{definition}
 
 \begin{definition}\label{def:jam}
For $C\in \{A,B\}$ and positive integers~$t$ and~$\tau$, $\Ejam^{ \YAB}(C,t,\tau)$ is the event that, for all $\tau'\in \{1,\ldots,\tau\}$,
$\YAB$ is $(C,t)$-jammed for $\tau'$.
$\Ejam^{\YAB}(t,\tau)$ is $\Ejam^{\YAB}(A,t,\tau) \cap \Ejam^{\YAB}(B,t,\tau)$.
\end{definition}
 
The heart of our proof will be Lemmas~\ref{lem:main}, which says that if $\Einit^T(t_0)$ occurs, then $T$ is likely to be $t_0$-jammed for all $\tau \ge 1$ (and hence never return to a state with no stuck balls). We will split this proof into two regimes. The first regime, $\tau \le \tauinit$, will be relatively simple and is covered by Lemma~\ref{lem:tau-init}; we are essentially just arguing that some of the stuck balls guaranteed by $\Einit^T(t_0)$ are likely not to have sent yet. The second regime, $\tau \ge \tauinit$, will be the heart of the proof and is covered by Lemma~\ref{lem:toget}.

\subsection{The \texorpdfstring{$\tau \le \tauinit$}{tau <= tau\_init} regime: Proving Lemma~\ref{lem:tau-init}}
 
\begin{lemma} \label{lem:tau-init}
Fix $\lambda$, $\eta$ and $\nu$ in $(0,1)$  and a $(\lambda,\eta,\nu)$-suitable send sequence~$\mathbf{p}$.
Let $\YAB$ be a two-stream externally-jammed process with arrival rate $\lambda$
and send sequence~$\mathbf{p}$.  Consider any $C\in \{A,B\}$.
Let 
$t_0$ be any positive integer and let $F^{\YAB}_{t_0}$ be a value of $\calF^{\YAB}_{t_0}$ such that $\Einit^{\YAB}(t_0)$ occurs. Then
the following two statements hold.
\begin{equation}\label{eq:deterministic_tauinit}
\mbox{
If $\calF^{\YAB}_{t_0} = F^{\YAB}_{t_0}$ occurs then
$\Ejam^{\YAB}(C,t_0,1)$ occurs, and}
\end{equation}
and

    \begin{equation}\label{eq:C-tau-init}
        \pr\big(\Ejam^{\YAB}(t_0,\tauinit) \mid \calF^{\YAB}_{t_0} = F^{\YAB}_{t_0}\big) \ge 49/50.
    \end{equation}
\end{lemma}

\begin{proof}

We first establish Equation~\eqref{eq:deterministic_tauinit}.  
We wish to show that
$\Ejam^{\YAB}(C,t_0,1)$ occurs which is the event that $\YAB$ is 
$(C,t_0)$-jammed for time~$1$, i.e.\ that $f(\stuckvect^{\YAB}(C,t_0)) \geq \zeta |\bins(0)| = \zeta |B_{I(0)-1}| = \zeta | B_{I_0}|$.
To establish this,  first note by the definition of $f$ (Definition~\ref{def:noise}) that
$$f(\stuckvect^{\YAB}(C,t_0)) \geq 
|\stuck_{\jmin}^{\YAB}(C,t_0)| p_{\jmin}.$$
This quantity can be deduced
from $F_{t_0}^{\YAB}$ and, since $F_{t_0}^{\YAB}$ guarantees that $\Einit^{\YAB}(t_0)$ happens,
is at least $\Cinit p_{\jmin}$.
To finish, we just need to show that $\Cinit p_{\jmin} \geq\zeta |B_{I_0}|$,
which follows from the choice of $\Cinit$ in~Definition~\ref{def:Cinit}.

We next establish Equation~\eqref{eq:C-tau-init}. Consider any integer~$\tau$
in the range
$1 \le \tau \le \tauinit$ and condition on $\calF^{\YAB}_{t_0} = F^{\YAB}_{t_0}$. Since $\Einit^{\YAB}(t_0)$ occurs under this conditioning, for each $C \in \{A,B\}$, we have $|\stuck^{\YAB}_{\jmin}(C,t_0)| \ge \Cinit$; we will argue that $\stuck^{\YAB}_{\jmin}(C,t_0)$ does not have time to empty before time $t_0+\tau$, so $\YAB$ is likely to be 
$(C,t_0)$-jammed for~$\tau$.  As before, observe that
    \begin{align}\nonumber
        f(\stuckvect^{\YAB}(C,t_0+\tau-1)) 
        &\ge p_{\jmin}|\stuck^{\YAB}_{\jmin}(C,t_0+\tau-1)|\\ \label{eq:C-tau-init-1}
        &\ge p_{\jmin} |\stuck^{\YAB}_{\jmin}(C,t_0) \cap \stuck^{\YAB}_{\jmin}(C,t_0+\tau-1)|.
    \end{align}
    For convenience, write $N_\tau = |\stuck^{\YAB}_{\jmin}(C,t_0) \cap \stuck^{\YAB}_{\jmin}(C,t_0+\tau-1)|$ for the size of the set in the right-hand side of~\eqref{eq:C-tau-init-1}.
    
    A ball is in $N_\tau$ if it is in $\stuck^{\YAB}_{\jmin}(C,t_0)$ and it does not send until time $t_0+\tau$ or later. Since messages send independently, conditioned on $\calF^{\YAB}_{t_0} = F^{\YAB}_{t_0}$, $N_\tau$ is binomially distributed. Let $\mu = \E(|N_\tau| \mid \calF^{\YAB}_{t_0}=F^{\YAB}_{t_0})$, and observe that $\mu = (1-p_{\jmin})^{\tau-1} |\stuck^{\YAB}_{\jmin}(C,t_0)| \ge (1-p_{\jmin})^{\tauinit}\Cinit$. 
     By the definition of~$\Cinit$ (definition~\ref{def:Cinit}), we have $\mu \ge 12\log(100\tauinit)$ and $\mu \ge 2\zeta |\bins(\tau-1)|/p_{\jmin}$. Using this fact, we apply the Chernoff bound of Lemma~\ref{lem:chernoff-small-dev} to $|N_\tau|$ with $\delta=1/2$ to obtain
    \begin{align*}
        \pr\big(|N_\tau| \ge \zeta |\bins(\tau-1)|/p_{\jmin} \mid \calF^{\YAB}_{t_0} = F^{\YAB}_{t_0}\big) &\ge \pr\big(|N_\tau| \ge \mu/2 \mid \calF^{\YAB}_{t_0} = F^{\YAB}_{t_0}\big)\\
        &\ge 1 - e^{-\mu/8} \ge 1 - 1/(100\tauinit).
    \end{align*}
    Hence by~\eqref{eq:C-tau-init-1} and the definition of $(C,t_0)$-jammedness, we have that for all $1 \le \tau \le \tauinit$,
    \[
        \pr\big(\mbox{$\YAB$ is $(C,t_0)$-jammed for time }\tau \mid \calF^{\YAB}_{t_0} = F^{\YAB}_{t_0}\big) \ge 1 - 1/(100\tauinit).
    \]
   Thus~\eqref{eq:C-tau-init}, and therefore also the result, follows by a union bound over all $C \in \{A,B\}$ and $\tau \in \{1,\dots,\tauinit\}$.
\end{proof}

\subsection{Properties of the block construction  }\label{sec:block}

\begin{observation}
\label{obs:block}
For $i\geq 1$ we have
 $u(i) = \Cblock^{i-1}$.
Also $|B_1|=1$. For
$i\geq 2$ we have $|B_i| 
=  \Cblock^{i-2}(\Cblock-1)
= u(i) (\Cblock-1)/\Cblock$
.

\end{observation}
\begin{proof}
For $i\geq 2$, we have 
$|B_i| 
= u(i) - u(i-1) = \Cblock^{i-1}-\Cblock^{i-2}$.

\end{proof}

\begin{lemma}\label{lem:I-bound}
{For all $\tau \geq \tauinit$},     
$|\bins(\tau)| > \log (\tau) / (2\Cblock^2\log( 1/\nu))$.
\end{lemma}

\begin{proof}
By definition, we have $|\bins(\tau)| = |B_{I(\tau)-1}|$. By Observation~\ref{obs:block} and since $\Cblock \ge 2$, 
it follows that $|\bins(\tau)| \ge \Cblock^{I(\tau)-3}$. Again by Observation~\ref{obs:block}, we have $u(I(\tau)) = \Cblock^{I(\tau)-1}$, and so
\begin{equation}\label{eq:I-bound-0}
        |\bins(\tau)| \ge u(I(\tau))/\Cblock^2.
\end{equation}
Our proof strategy will be to bound $I(\tau)$ below in order to apply~\eqref{eq:I-bound-0}. To do this, we will first bound $\tau$ above in terms of $u(I(\tau))$, and then convert this into a lower bound on $u(I(\tau))$ in terms of $\tau$. By the definitions of~$I(\tau)$ and~$\tau_i$, we have
\[
\tau < \tau_{I(\tau)-I_0} = \Cblock\sum_{k=1}^{I_0} (I(\tau)-k+1)\lceil W_k \rceil + \Cblock\sum_{k=I_0+1}^{I(\tau)} (I(\tau)-k+1)\lceil W_k \rceil.
\]

We will argue that the second sum is larger than the first. To see this, 
let $Q = \max_{k=1}^{I_0} \lceil W_k \rceil$.
Then the first sum is at most 
$\Cblock I_0 I(\tau)    Q$. 
To lower-bound the second sum, just consider $k$ up as high as $k = (I(\tau)/2)+1$
then  each term (for each value of $k$) is at least  
$\Cblock I(\tau)/2$
and the number of terms is at least $(I(\tau)/2 )  - I_0$.
The definition of $\tauinit$ insures that
$I(\tauinit) \geq 2 I_0 (2 Q + 1)$. Thus,
for every $\tau \geq \tauinit$,
the second term  is at least as large as the first.
 
It   follows that
\begin{align*}
\tau &< 2\Cblock\sum_{k=I_0+1}^{I(\tau)} (I(\tau)-k+1) \lceil W_k \rceil = 2\Cblock\sum_{k=I_0+1}^{I(\tau)}\sum_{m=I_0+1}^k \lceil W_m \rceil = 2\Cblock\sum_{k=I_0+1}^{I(\tau)}\bigg\lceil \sum_{j=\ell(I_0+1)}^{u(k)} (1/p_j) \bigg\rceil.
    \end{align*}
We bound the right-hand side above by observing that \[
\big\{u(k) \colon I_0+1 \le k \le I(\tau)\big\} \subseteq \big\{\ell(I_0+1), \ell(I_0+1)+1, \dots, u(I(\tau))\big\},
    \]
and so
\begin{equation}\label{eq:I-bound-tau-bound}
\tau < 2\Cblock\sum_{a = \ell(I_0+1)}^{u(I(\tau))} \bigg\lceil \sum_{j=\ell(I_0+1)}^a (1/p_j)\bigg\rceil \le 4\Cblock\sum_{a = \ell(I_0+1)}^{u(I(\tau))} \sum_{j=\ell(I_0+1)}^a (1/p_j).
    \end{equation}
By the definition of $I_0$, 
for all $i \ge I_0+1$ and all $j \in B_i$ we have $p_j > \nu^j$. It follows that for all $a \ge \ell(I_0+1)$, 
\[
        \sum_{j=\ell(I_0+1)}^a \frac{1}{p_j} < \sum_{j=\ell(I_0+1)}^a \left(\frac{1}{\nu}\right)^j < \left( \frac{1}{\nu}\right)^a\sum_{j=0}^\infty  
         \nu^j = \frac{1}{1- \nu}
        \left(\frac{1}{ \nu}\right)^a.
    \]
Hence by~\eqref{eq:I-bound-tau-bound},
\begin{equation}\label{eq:I-bound-u-bound}
\tau < 
\frac{4\Cblock}{1- \nu}
\sum_{a = \ell(I_0+1)}^{u(I(\tau))} \left(\frac{1}{\nu}\right)^a < 
4\Cblock\frac{ \left( \frac{1}{\nu} \right)^{u(I(\tau))} }
{1- \nu}\sum_{a=0}^\infty 
 \nu^a = 
\frac{4\Cblock}{\left(1- \nu\right)^2}
\left( \frac{1}{\nu} \right)^{u(I(\tau))}.
    \end{equation}
    
Solving for $u(I(\tau))$ now yields
\[
u(I(\tau)) > \log_{1/\nu} (\tau) + 2\log_{1/\nu}\Big(\frac{1- \nu}{2\sqrt{\Cblock}}\Big).
\]

Since $\tau \ge \tauinit$
and $\tauinit \geq (2\Cblock/(1-\nu))^4$,
it follows that $u(I(\tau)) > \tfrac{1}{2}\log_{1/\nu}\tau$. By~\eqref{eq:I-bound-0}, it follows that
\[
|\bins(\tau)| \ge \frac{1}{\Cblock^2}u(I(\tau)) > \frac{1}{2\Cblock^2\log(1/ \nu)}\log \tau,
\]
as required.
\end{proof}
 
\subsection{Noise}

\begin{lemma}\label{lem:one_onesided}
Let $\YABfull$ be a two-stream externally-jammed process with arrival rate $\lambda$
and send sequence~$\mathbf{p}$  with $p_0=1$.
Let $t \ge 0$ and let 
$C\in \{A,B\}$.
Let $\boldx$ be 
any possible value of    $\stuckvect^{\YAB}(C,t)$.
Then  
$$
\pr( 
\stucksend^{\YAB}(C,t+1)=\emptyset \mid \stuckvect^{\YAB}(C,t)=\boldx) \leq \exp(- f(\boldx)/3).
$$
\end{lemma}

\begin{proof}
For brevity, write $q$ for the probability in the statement of the lemma, and let $\Psi$ be the number of balls in $\stuck^T(C,t) \cap \send^T(t+1)$. Then since any ball in $\stucksend^T(C,t+1)$ is either born at time $t+1$ or in $\Psi$, we have
\begin{equation}\label{eq:one_onesided-0}
    q = e^{-\lambda/2}\pr(\Psi = 0 \mid \stuckvect^T(C,t) = \boldx).
\end{equation}

Conditioned on $\stuckvect^T(C,t) = \boldx$, $\Psi$ is a sum of independent Bernoulli variables with mean $\mu = f(\boldx) - \lambda$.
If one of these Bernoulli variables has parameter~$p$ then the probability that it is~$0$
is $1-p \leq \exp(-p)$. Thus the probability that they are all~$0$
is at most $\exp(-\mu)$. We therefore have
\[
    \pr(\Psi = 0 \mid \stuckvect^T(C,t) = \boldx) \le e^{-\mu}.
\]
It follows by~\eqref{eq:one_onesided-0} that 
$q\leq e^{-f(\boldx) + \lambda/2}   < e^{-f(\boldx)/3}$ as required (since $f(\boldx) \geq \lambda$ by definition). 
\end{proof}

\section{Towards a proof of Theorem~\ref{thm:technical}: Jamming causes 
 a product distribution of 
stuck balls in a two-stream externally-jammed process}\label{sec:rev}

 The purpose of this section is to prove the following lemma, which will be our main tool to analyse the two-stream process in the proof of Lemma~\ref{lem:toget} (and thereby Theorem~\ref{thm:technical}).
 
 \begin{restatable}{lemma}{lemnewdombinoms}\label{lem:newdombinoms}\label{lem:interface}
 
 Fix $\lambda$, $\eta$ and $\nu$ in $(0,1)$  and a $(\lambda,\eta,\nu)$-suitable send sequence~$\mathbf{p}$ with $p_0=1$.
 Let $Y^A$ and $Y^B$ be independent externally-jammed processes with arrival rate $\lambda/2$ and send sequence $\mathbf{p}$  and consider the two-stream externally-jammed process $\YAB=\YAB(Y^A,Y^B)$. Consider any integer~$t_0 \geq  \jmin$. 
Let $F^{\YAB}_{t_0}$ be a value of $\calF^{\YAB}_{t_0}$ 
such that $\Einit^{\YAB}(t_0)$  happens. Consider any integer $\tau \geq \tauinit$.
Then  there is a coupling of
\begin{itemize}
    \item      $\YAB$
 conditioned on $\calF^{\YAB}_{t_0} = F^{\YAB}_{t_0}$, 
 \item a sample $\{Z_j^{A} \mid j \in \bins(\tau)\}$ where each $Z_j^{A}$ is
 chosen independently from a Poisson distribution with mean $\lambda/(4 p_j)$, and 
 \item a sample $\{Z_j^{B} \mid j \in \bins(\tau)\}$ where each $Z_j^{B}$ is chosen from a Poisson distribution with mean $\lambda/(4 p_j)$ and these are independent of each other but not of the $\{Z_j^A\}$ values
 \end{itemize}
  in   such a way that  
 at least one of the following happens:
 
   \begin{itemize}
     \item $\Ejam^{\YAB}(t_0,\tau)$ does not occur, or
    \item for all $j\in \bins(\tau)$, 
    $|\stuck_j^{\YAB}(A,t_0+\tau)  | \geq Z_j^A$
    and $|\stuck_j^{\YAB}(B,t_0+\tau)  | \geq Z_j^B$.
 \end{itemize}
 \end{restatable}
 
In order to prove Lemma~\ref{lem:interface}
we need to study several stochastic processes, starting with a random-unsticking process in the next section.

\subsection{The definition of a random-unsticking process}
\label{sec:defRandom}
\label{sec:particularRandom}

Let $Y$ be an externally-jammed process with arrival rate~$\lambda$ and send sequence $\mathbf{p}$.
Let $t_0$ be a fixed positive integer. 
We will define a \emph{random-unsticking process}~$\Random=\Random(Y,t_0)$ by redefining the ``stuck'' and ``unstuck'' states of $Y$, so that balls in $\Random$ become unstuck independently at random with probability depending on $t_0$. Later, in Lemma~\ref{lem:domcouplenew}, we will show that $R$ and $T$ can be coupled in such a way that as long as $T$ stays jammed, if a ball is stuck in $R$ then it is also stuck in $T$; this will allow us to analyse $R$ rather than $T$ in proving Lemma~\ref{lem:newdombinoms}.

Apart from this redefinition, which explained below, everything about $\Random$ is the same as~$Y$: 
for all $t \geq 0$, $\balls^R(t) = \balls^T(t)$;
for all $t\geq 0$ and $j\geq 1$, $b_j^{\Random}(t) = b_j^{Y}(t)$;
for all $t\geq 1$ and $j\geq 0$, $b^{\prime\Random}_j(t) = b^{\prime Y}_j(t)$;
for all $t\geq 1$ $\send^{\Random}(t) = \send^{Y}(t)$;
and for $t \geq 0$ and all $\beta \in \balls^T(t)$, $\birth^R(\beta) = \birth^Y(\beta)$.

We now formalise the new ``stuck'' and ``unstuck'' states in $\Random$.
The set
$\stuck^{\Random}(t)$ will be the set of all
balls in $\balls^{\Random}(t)$  
that are stuck 
in the process~$\Random$ at time~$t$
and $\stuck_j^{\Random}(t) = \stuck^{\Random}(t) \cap b_j^T(t)$.
As the reader will expect,
$\unstuck^{\Random}(t) = \balls^{\Random}(t) \setminus \stuck^{\Random}(t)$
so for each~$t$ we only need to define $\stuck^{\Random}(t)$ or $\unstuck^{\Random}(t)$
and the other is defined implicitly.
 These variables are defined as follows.
First, let
$\stuck^{\Random}(0) = \emptyset$ and 
$$\unstickProb(t) =  
\begin{cases}
1, &\mbox{if $t\leq t_0$}\\
e^{-\zeta 
|\bins(t-t_0)| 
/16}, &\mbox{otherwise.}
\end{cases}
 $$
Then for any $t\geq 1$ we have the following definitions:
\begin{itemize}
\item ${\newstuck}^{\Random}(t) = \stuck^{\Random}(t-1) \cup b^{\prime\Random}_0(t)$.
\item $\stucksend^{\Random}(t) = {\newstuck}^{\Random}(t) \cap \send^{\Random}(t)$.
\item For each $\beta \in \balls^{\Random}(t)$,
  $\unstickIndicator(\beta,t)$
is  an independent Bernoulli random variable which is~$1$ with 
probability $\unstickProb(t)$ and $0$~otherwise.
\item  Then
$\unstuck^{\Random}(t) = {\unstuck}^{\Random}(t-1) \cup 
\{\beta \in \stucksend^{\Random}(t) \mid \unstickIndicator(\beta,t)=1
\}$.
\end{itemize}

 We take   $\calF^{\Random}_t$ to be
$\calF^{Y}_t$ 
 along with the values of the random variables $\unstickIndicator(\beta,t')$
 for 
 $t' \leq t$  and $\beta\in \balls^{\Random}(t')$ 
so that all information about the evolution
of the process~$\Random$, including the  
value of $\stuck^{\Random}(t)$, is captured.

\subsection{Defining a reverse  random-unsticking  process}\label{sec:rev-def}

Given an externally-jammed process $Y$ with arrival rate~$\lambda$,
a positive integer~$t_0$, a random-unsticking process
$R=R(Y,t_0)$, and positive integers~$j$ and $\tauend$, we
wish to  find a lower bound on 
$|\stuck_j^\Random(t_0+\tauend)|$ in order to prove Lemma~\ref{lem:newdombinoms}. To do this, we'll  define the 
notion of a 
reverse  random-unsticking process, which is essentially a time-reversal of a random-unsticking process; we will make this relation formal later, in Lemma~\ref{lem:timerev}.
The reverse random-unsticking process
$\revRandom= \widetilde{\Random}(Y,t_0,\tauend)$ 
  will be easier to 
study than $\Random$~itself, allowing us to
deduce the lower bounds on
$|\stuck_j^{\Random}(t_0+\tauend)|$ that we need. 
Note that the basic idea of using a time-reversal in the context of backoff processes
is due to Aldous~\cite{Aldous}.
Here we require independent lower bounds for different~$j$,
so it is a lot more complicated.

The process~$\revRandom$ will run for
steps $1,\ldots,\tauend$.
Each step consists of   (II) sending, 
(III) adjusting the bins, and
(IV) random unsticking. 
(We start counting from (II) so that the indices line up with the parts
of the steps of an externally-jammed process. Here there is no 
step initialisation, since there are no births. On the other hand, there is random unsticking.)
Let $\Jmax = u(I(\tauend)-1) + \tauend$.
The bins in the process~$\revRandom$ 
will be the integers $j\in [\Jmax]$.
For $j\in [\Jmax]$,  
the set $b_j^{\revRandom}(\tau)$ is the population of bin~$j$
just after (all parts of) the $\tau$'th step. 
For convenience, we define $b_{\Jmax+1}^{\revRandom}(\tau)=\emptyset$.
Also, $\balls^{\revRandom}(\tau) = \cup_{j=1}^{\Jmax} b_j^{\revRandom}(\tau)$.

The process $\revRandom$  is initialised as follows.
For each $j\in [\Jmax]$, the number $\tilde{x}_j$ 
of balls that start in $b_j^{\revRandom}(0)$ is
chosen independently from a Poisson distribution with mean~$\lambda/p_j$.
These balls are given unique names $(\revRandom,j,1),\ldots,(\revRandom,j,\tilde{x}_j)$. 
Step~$\tau$ is defined as follows.

\begin{itemize}

\item Part (II) of step $\tau$: For   
$j \in [\Jmax]$, 
all balls in $b^{\revRandom}_j(\tau-1)$ 
send independently with probability $p_j$. 
We use $\send^{\revRandom}(\tau)$ for the set of balls that send at time~$\tau$.

\item Part (III) of step $\tau$: 
For all $j\in [\Jmax]$, 
$$b_j^{\revRandom}(\tau) = (b_{j+1}^{\revRandom}(\tau-1)\cap \send^{\revRandom}(\tau))
\cup
(b_j^{\revRandom}(\tau-1) \setminus \send^{\revRandom}(\tau)).
$$

\item Part (IV) of step $\tau$:   For each $\beta \in  \balls^{\revRandom}(\tau-1)$,
$\unstickIndicator(\beta,\tau)$ is an independent Bernoulli random variable,
which is~$1$ with probability $\unstickProb(t_0+\tauend - \tau+1)$
and~$0$ otherwise.

\end{itemize}
 
Note that balls in $b_1^{\revRandom}(\tau-1)$ that send at step~$\tau$
are not part of $\balls^{\revRandom}(\tau)$ --- instead, they leave the system at step~$\tau$. 
Note also that we do not define ``stuck'' and ``unstuck'' states for $\revRandom$ --- it will actually be more convenient to work with the $\unstickIndicator$ variables directly. (These will later be coupled to the corresponding variables in a random-unsticking process $\revRandom$; as the definition suggests, unsticking events at times $1,2,\dots,\tauend$ in $\revRandom$ will correspond to unsticking events at times $t_0+\tauend,t_0+\tauend-1,\dots,t_0+1$ in $\Random$.)

\subsection{ Defining trajectories in a random-unsticking process}\label{sec:traj-def}

Let $Y$ be an externally-jammed process with arrival rate~$\lambda$ and let $t_0$ be 
a positive integer. Consider a random-unsticking process $\Random = \Random(Y,t_0)$.

We will only be concerned with balls~$\beta$
with $\birth^{\Random}(\beta) \geq t_0+1$.  For $t\geq t_0+1$
consider a ball with  name $\beta=(Y,t,x)$ 
for
some~$x$. The \emph{ball trajectory}~$\trajectory$
of the ball~$\beta$ 
(up to time step $t_0+ \tauend$ for a positive integer $\tauend$)
contains the following information:
\begin{itemize}
\item $\tbirth(\trajectory) = \birth^{\Random}(\beta)$.

\item The bin $J(\trajectory)$ such that $\beta \in b^{\Random}_{J(\trajectory)}(t_0+\tauend)$.

\item Positive integers $N_1(\trajectory),\ldots,N_{J(\trajectory)}(\trajectory)$ adding up to 
$t_0+\tauend - \tbirth(\trajectory) +1$. For each $j\in [J(\trajectory)]$,
$N_j(\trajectory)=  |\{
t' \mid 
\tbirth(\trajectory) \leq t' \leq t_0+\tauend,\ 
\beta \in b_j^{\Random}(t')
\}|$. 

\item For each 
integer~$t'$ in the range $\tbirth(\trajectory) \le t' \le t_0+\tauend$, 
the indicator variable $\unstickIndicator(\trajectory,t')=\unstickIndicator(\beta,t')$.
\end{itemize}

Note that any choice of positive integers $N_1(\trajectory),\ldots,N_{J(\trajectory)}(\trajectory)$
with the right sum leads to a valid (ball) trajectory, and so does any 
sequence of indicator random variables 
$\unstickIndicator(\trajectory,t')$. 
Thus, 
for all fixed~$\lambda$, $t_0$ and $\tauend$, all of the information in trajectory~$\trajectory$ is
captured by $J(\trajectory)$,
$\tbirth(\trajectory)$,  
$\{N_j(\trajectory) \mid 1\leq j \leq J(\trajectory)\} $ and 
$\{ \unstickIndicator(\trajectory,t') \mid \tbirth(\trajectory) \leq \tau \leq t_0+\tauend\}$.
Let $\TrajectorySet^{\Random}(t_0,\tauend,t,J)$ 
be the set of trajectories up to $t_0+\tauend$ with $\tbirth(\trajectory)=t$ and $J(\trajectory)=j$.

Given a trajectory~$\trajectory$, let
$S(\trajectory) = \{ 
\tbirth(\trajectory) + \sum_{k=1}^c N_k(\trajectory) \mid c\in \{0,\ldots,J(\trajectory)-1\}
\}$. We will use the fact that $S(\trajectory)$ is 
the set of time steps up to time $t_0+\tauend$ at which 
balls 
following trajectory~$\trajectory$ send.

We can calculate the distribution of the number of balls following
a given trajectory $\trajectory$. 
 This is a Poisson random variable whose mean, $\mu^{\Random}(\trajectory)$, is 
 given by $\mu^{\Random}(\trajectory) = F_1 F_2 F_3$ where $F_1$, $F_2$ and $F_3$ are defined as follows:
 \begin{align*}
F_1 &= \lambda (1-p_{J(\trajectory)})^{N_{J(\trajectory)}(\trajectory)-1};\qquad\qquad
F_2 = \prod_{j=1}^{J(\trajectory)-1} p_j(1-p_j)^{N_j(\trajectory)-1};\\
F_3 &= \prod_{t'=\tbirth(\trajectory)}^{t_0+\tauend} \unstickProb(t')^{\unstickIndicator(\trajectory,t')} (1-\unstickProb(t'))^{1-\unstickIndicator(\trajectory,t')}.
\end{align*}

 \subsection{Defining trajectories in a 
 reverse random-unsticking process \texorpdfstring{$\revRandom$}{R}}

Let $Y$ be  an externally-jammed process with arrival rate~$\lambda$,
let $t_0$ be a positive integer and let $\tauend$ be a positive integer.
Consider a reverse random-unsticking process 
$\revRandom= \widetilde{\Random}(Y,t_0,\tauend)$

We will only be concerned with balls in $\revRandom$ that leave the process
by sending from bin~$1$ at 
some step $\tau \in \{1,\ldots,\tauend\}$. 
We next define the trajectory~$\revT$ 
of such a ball with  name $\beta=(\revRandom,j',x)$ in~$\revRandom$.
The trajectory~$\revT$ contains the following information:
\begin{itemize} 
\item The bin $J(\revT)=j'$, the bin into which $\beta$ is  placed in the initialisation,
so $\beta \in b_{J(\revT)}^{\revRandom}(0)$.

\item The time $\tauleave(\revT)$ at which $\beta$ leaves the process by sending from bin 1,
so $\beta \in b_{1}^{\revRandom}(\tauleave(\revT)-1)$
but $\beta \notin b_1^{\revRandom}(\tauleave(\revT))$.

\item Positive integers $N_1(\revT),\ldots,N_{J(\revT)}(\revT)$,
adding up to $\tauleave(\revT)$.
For each $j \in [J(\revT)]$, 
$N_j(\revT) = |\{
\tau' \mid 
0 \leq \tau' \leq \tauleave(\revT)-1,\ \beta \in b_j^{\revRandom}(\tau')
\}|
$.

\item For each time step $\tau'$ in the range
$1\leq \tau' \le \tauleave(\revT)$,
the indicator variable $\unstickIndicator(\revT,{\tau'}) = \unstickIndicator(\beta,\tau')$.
\end{itemize}

Note that any choice of positive integers $N_1(\revT),\ldots,N_{J(\revT)}(\revT)$
with the right sum leads to a valid trajectory, and so does any 
sequence of indicator random variables $\unstickIndicator(\revT,{\tau})$.
Thus for all fixed~$\lambda$, $t_0$ and $\tauend$, 
all of the information in trajectory~$\trajectory$ is
captured by $J(\revT)$,
$\tauleave(\revT)$,  
$\{N_j(\revT) \mid 1\leq j \leq J(\revT)\} $ and 
$\{ \unstickIndicator(\revT,\tau) \mid 1 \leq \tau \leq \tauleave(\revT)\}$.
Let $\TrajectorySet^{\revRandom}(t_0,\tauend,\tau,J)$ 
be the set of trajectories with $\tauleave(\revT)=\tau$ and $J(\revT)=J$.

Given a trajectory~$\revT$, let
$\revS(\revT) = \{ 
  \sum_{k=\revc}^{J(\revT)} N_k(\revT) \mid \revc\in \{1,\ldots,J(\revT)\}
\}$. We will use the fact that $\revS(\revT)$ is
the set of time steps at which 
balls following trajectory~$\revT$ send.

In the   process $\revRandom$, the distribution for the number of balls following trajectory~$\revT$
is a Poisson random variable whose mean, $\mu^{\revRandom}({\revT})$, is 
given by $\mu^{\revRandom}({\revT})= \tilde{F}_1 \tilde{F}_2 \tilde{F}_3$ 
where $\tilde{F}_1$, $\tilde{F}_2$ and $\tilde{F}_3$ are
defined
as follows.

\begin{align*} 
\tilde{F}_1 &=
\frac{\lambda}{p_{J(\revT)}}; \qquad\qquad
\tilde{F}_2 = \prod_{j=1}^{J(\revT)} p_{j}(1-p_{j})^{N_{j}({\revT})-1}; \\
\tilde{F}_3 &=
\prod_{\tau=1}^{\tauleave(\revT)} \unstickProb(\tauend-\tau+t_0+1)^{\unstickIndicator(\revT,{\tau})} \big(1-\unstickProb(\tauend-\tau+t_0+1)\big)^{1-\unstickIndicator(\revT,{\tau})}.
\end{align*}

\subsection{\bf A probability-preserving bijection from trajectories of a
random-unsticking process\texorpdfstring{$~\Random$}{} to those of the corresponding reverse
random-unsticking process\texorpdfstring{$~\revRandom$}{}}\label{sec:coupling}

Let $Y$ and $Y'$ be externally-jammed processes with arrival rate~$\lambda$ and send sequence $\mathbf{p}$,
and let $t_0$ 
and $\tauend$ be 
positive integers.
Consider a random-unsticking process $\Random = \Random(Y,t_0)$
and  a reverse random-unsticking process 
$\revRandom= \widetilde{\Random}(Y',t_0,\tauend)$  We now set out the natural bijection from trajectories of $\Random$ to trajectories of $\revRandom$ that will later form the basis of our coupling.

\begin{definition}\label{def:time-rev-bi}
    The \emph{time-reversal bijection} $\pi$ from $\TrajectorySet^{\Random}(t_0,\tauend,t,J)$ to $\TrajectorySet^{\revRandom}(t_0,\tauend, \tauend - t+t_0+1,J)$ is defined as follows for all $\trajectory \in  \TrajectorySet^{\Random}(t_0,\tauend,t,J)$:
    \begin{itemize}
        \item $\tauleave(\pi(\trajectory)) = \tauend - \tbirth(\trajectory)+t_0+1$;
        \item $J(\pi(\trajectory))=J$;
        \item for all $k \in [J]$, $N_k(\pi(\trajectory)) = N_k(\trajectory)$;
\item for all 
$t'$ in the range
$\tbirth(B) \leq t' \leq t_0 + \tauend$,
$\unstickIndicator(\pi(\trajectory),\tauend - t' + t_0+1) = \unstickIndicator(\trajectory,t')$.
    \end{itemize}
\end{definition}

We now pause for a moment to observe that $\pi$ is a valid map into $\TrajectorySet^{\revRandom}(t_0,\tauend, \tauend - t+t_0+1,J)$. Indeed, we have $\tbirth(\trajectory) \ge t_0+1$, so $\tauleave(\pi(B)) \le \tauend$ as required. In order for $J(\pi(\trajectory))=J$ to be a bin of the process $\revRandom$
we need $J $   to be in $[\Jmax]$;
this is fine since the definition of $J(\trajectory)=J$ ensures that $J(\trajectory) \leq \tauend$
and the definition of
$\Jmax$ ensures that $1\leq \tauend \leq \Jmax$ so $J(\trajectory) \leq \Jmax$. We also have from the definitions that
\[
    \sum_{j=1}^J N_j(\trajectory) = \tauend - \tbirth(\trajectory)+t_0+1 = \tauleave(\trajectory) = \sum_{j=1}^J N_j(\pi(\trajectory)),
\]
as required. Finally, the values of $\tau$ for which we define $\unstickIndicator(\pi(\trajectory),\tau)$ range from $\tauend-\tbirth(\trajectory)+t_0+1 = \tauleave(\pi(\trajectory))$ down to $1$, as required. Thus the image of $\pi$ is indeed contained in $\TrajectorySet^{\revRandom}(t_0,\tauend, \tauend - t+t_0+1,J)$, and it is clear that $\pi$ is a bijection as claimed.

The following lemma formalises our intuition that $\revRandom$ is a time reversal of $\Random$, and will later form the basis for our coupling between them.

\begin{lemma}\label{lem:timerev}
Let $Y$ and $Y'$ be externally-jammed processes with arrival rate~$\lambda$ and send sequence $\mathbf{p}$ and let $t_0$ 
and $\tauend$ be 
positive integers. Consider a random-unsticking process $\Random = \Random(Y,t_0)$
and  a reverse random-unsticking process 
$\revRandom= \widetilde{\Random}(Y',t_0,\tauend)$.
Let $t$ 
be an integer satisfying $t\geq t_0+1$.
Let $J$ be a positive integer.
Let $\pi$ be the time-reversal bijection from  
$\TrajectorySet^{\Random}(t_0,\tauend,t,J)$ to $\TrajectorySet^{\revRandom}(t_0,\tauend,
\tauend - t+t_0+1,J)$.
Then for any $\trajectory \in   \TrajectorySet^{\Random}(t_0,\tauend,t,J)$, we have
$\mu^{\Random}(\trajectory) = \mu^{\revRandom}(\pi(\trajectory))$ and $\revS(\pi(\trajectory)) = \{\tauend - t + t_0+1 \mid t \in S(\trajectory)\}$.

\end{lemma} 
\begin{proof}
It is straightforward to use the definitions of~$\pi$, $\mu^{\Random}$ and $\mu^{\revRandom}$ to see
 that 
$\mu^{\Random}(\trajectory) = \mu^{\revRandom}(\pi(\trajectory))$. The only subtlety 
is checking that the   probabilities of  the unsticking indicator variables correspond. For this,
note from the definition
of~$\pi$
that for $\tau \in \{1,\ldots,\tauleave(\pi(\trajectory))\}$,
the probability that $\unstickIndicator(\pi(\trajectory),\tau)=1$ is
$\unstickProb(t')$ where $\tau = \tauend - t' + t_0+1$.
So this probability is
$\unstickProb(\tauend - \tau + t_0+1)$, as required.

To see that  $\revS(\pi(\trajectory)) = \{\tauend - t + t_0+1 \mid t \in S(\trajectory)\}$,
note that  if we take 
$ 
t = \tbirth(\trajectory) + \sum_{k=1}^c N_k(\trajectory) $
in $S(\trajectory)$
(for some $c\in \{0,\ldots,J(\trajectory)-1\}$)
 then 
 the quantity 
 $\tauend - t + t_0+1$ is equal to 
 $\tauend - \tbirth(\trajectory) - \sum_{k=1}^c N_k(\trajectory) + t_0+1$
 which, using the  fact that the sum of all of the $N_k(\trajectory)$'s
 is $\tauend - \tbirth(\trajectory) + t_0+1$, is
  $$\tauend - \tbirth(\trajectory) - (\tauend - \tbirth(\trajectory) + t_0+1) + \sum_{k=c+1}^{J(\trajectory)} N_k(t) + t_0+1.$$
 Clearly, this is
  $ \sum_{k=\revc}^{J(\revT)} N_k(\revT)$
 for $\revc = c+1$, which is in $\revS(\pi(\trajectory))$.

\end{proof}
   
\subsection{The set of balls that ``fill'' bin~\texorpdfstring{$j$}{j}} \label{sec:fill-def}
 
 Let $Y$ be an externally-jammed process with arrival rate~$\lambda$ and let $t_0$ be 
a positive integer. Consider a random-unsticking process $\Random = \Random(Y,t_0)$. 
 The following definitions will be convenient. 
 
  For   any positive integer~$j$ in block~$B_i$ and any positive integer~$\tau$, 
  we define $\Fill_j^{\Random}(\tau)$
  to be the subset of balls   $\beta \in \stuck_j^{\Random}(t_0+\tau)$
such that
\[
    \birth^{\Random}(\beta) \geq \max\Big\{t_0+1,t_0+\tau- \Cslack {\sum_{k=1}^i \lceil W_k \rceil }\Big\}.
\]

 We make a similar definition for the process $\YAB=\YABfull$.
For any $C\in \{A,B\}$,  any positive integer~$j$ in $B_i$, and any positive integer~$\tau$,  $\Fill_j^{ \YAB}(C,\tau)$
is the subset of balls   $\beta \in \stuck_j^{\YAB}(C,t_0+\tau)$
such that
\[
    \birth^{\YAB}(\beta) \geq \max\Big\{t_0+1,t_0+\tau- \Cslack {\sum_{k=1}^i \lceil W_k \rceil}\Big\}.
\]

It will be convenient to think about $\Fill_j^{\Random}(\tau)$ in terms of trajectories,
so we make the following equivalent definition
for  $\Fill_j^{\Random}(\tau)$: 
For   any positive integer~$j$ in block~$B_i$ and any positive integer~$\tauend$, 
  we define $\Fill_j^{\Random}(\tauend)$
  to be 
 the set of balls that take trajectories~$\trajectory$
satisfying the following:
  
\begin{enumerate}[(i)]
\item $J(\trajectory)=j$.
\item $\tbirth(\trajectory) \ge t_0+\tauend- \Cslack  {\sum_{k=1}^i \lceil W_k \rceil}$.
\item $\tbirth(\trajectory) \geq t_0+1$.
\item for all $t'\in S(\trajectory)$, $\unstickIndicator(\trajectory,t')=0$.
\end{enumerate}

Finally, we define  a corresponding notion $\Fill_j^{\revRandom}(\tauend)$ for a reverse random-unsticking process.
We first define 
 $\Fill_j^{\revRandom}(\tauend)$  in terms of trajectories.
Let $Y$ be an externally-jammed process with arrival rate~$\lambda$ and let $t_0$ 
and $\tauend$ be 
positive integers. Consider a reverse random-unsticking  process $\revRandom = \revRandom(Y,t_0,\tauend)$.
Let $j$ be a positive integer in block~$B_i$.
 Then $\Fill_j^{\revRandom}(\tauend)$ is defined to be the set of balls that take trajectories~$\revT$
satisfying the following.
  
\begin{enumerate}[(i)]
\item $J(\revT)=j$.
\item $\tauleave(\revT) \leq  1 + {\Cslack} {\sum_{k=1}^i \lceil W_k \rceil}$.
\item $\tauleave(\revT) \leq \tauend$.
\item for all $\tau'\in \revS(\revT)$, $\unstickIndicator(\revT,\tau')=0$.
\end{enumerate}

 Equivalently, 
 $\Fill_j^{\revRandom}(\tauend)$ is  the set of balls~$\beta=(\revRandom,j,x')$
satisfying the following.
  
\begin{enumerate}[(I)]
\item $\beta \in b_{j}^{\revRandom}(0)$.
\item  
For some $\tau' \leq 1 + {\Cslack} {\sum_{k=1}^i \lceil W_k \rceil}$,
$\beta \in b_1^{\revRandom}(\tau'-1)$ and $\beta$ sends at time $\tau'$. 
\item This $\tau'$ also satisfies
$\tau' \leq \tauend$. 
\item $\beta$ does not send on any step 
$\tau \in \{1,2,\ldots, 1 + {\Cslack} {\sum_{k=1}^i \lceil W_k \rceil}\}$
with $\unstickIndicator(\beta,\tau)=1$.
\end{enumerate}

\subsection{Bounding \texorpdfstring{$\Fill_j^{\revRandom}(\tauend)$}{Fill in a reverse random-unsticking process}}
 
 \begin{lemma}\label{lem:revfill} 
  Let $Y$ be an externally-jammed process with arrival rate~$\lambda$ and send sequence $\mathbf{p}$.
Consider any integer~$t_0 \geq \jmin$. 
Let~$\tauend\geq \tauinit$ be an   integer.
Then there is a coupling of
\begin{enumerate}[(i)]
\item a reverse random-unsticking process $\revRandom(Y,t_0,\tauend)$   with
\item a sample $\{Z_j^- \mid j \in \bins(\tauend)\}$ where each $Z_j^-$ is chosen independently from a Poisson distribution with mean $\lambda/(4 p_j)$ 
\end{enumerate}
in such a way that, for every integer $j\in B_{I(\tauend)-1}$,  
  $|\Fill_j^{\revRandom}(\tauend)| \ge Z_j^-$.
\end{lemma}
 
\begin{proof}
Let $\revRandom = \revRandom(Y,t_0,\tauend)$. We will use the definition of   $\Fill_j^{\revRandom}(\tauend)$ 
without trajectories. 
To simplify the notation, let $i= I(\tauend)-1$, so that $\bins(\tauend) = B_i$.
 {Since $\tauend \geq \tauinit$, and since the definition of $\tauinit$ (Definition~\ref{def:tauinit}) ensures
that $I(\tauinit) \geq I_0+3$, we have $i\geq I_0+2$.}

 {We will be considering $j\in B_i$.  Note that this }ensures that $j\leq u(i)$ 
and this is at most~$\Jmax$, by the definition of~$\Jmax$ in Section~\ref{sec:rev-def}. 
Thus, by the initialisation of the  process~$\revRandom$, the number of balls in $b_j^{\revRandom}(0)$ is   Poisson with mean~$\lambda/p_j$ (independently of any other $j$). Since each $\Fill_j^{\revRandom}(\tauend)$ is determined by the trajectories of balls in $b_j^{\revRandom}(0)$, and since all balls behave independently in the process~$\revRandom$, the random variables $\Fill_j^{\revRandom}(0)$ are already distributed as independent 
Poisson random 
variables as in (ii).
So to finish, we just need to show that, for all $j \in B_i$ and all balls $\beta \in b_{j}^{\revRandom}(0)$,
the probability that (II), (III) and (IV) happens is at least~$1/4$.
To do this, we'll show that the probability that (II) or (III)  fails is at most~$1/2$
and that the probability that (IV)  fails is at most $1/4$ (then we are finished by a union bound adding the failure probabilities~$1/2$ and~$1/4$).
 
For (IV),
it suffices to show that  for any $\tau \in \{1,\ldots,
 1 + {\Cslack} { \sum_{k=1}^i \lceil W_{k} \rceil}\}$,
$\unstickProb(t_0+ \tauend - \tau+1) \leq 1/(4 j)$. Using a union bound over
the   (at most~$j$) steps when $\beta$ sends, this implies that the probability that 
$\beta$ ever 
has $\unstickIndicator(\beta,\tau)=1$
on one of these steps~$\tau$  is at most~$1/4$.
Towards this end, first note from the definitions of $I(\cdot)$ and $\tau_{i'}$ that 
\begin{equation}\label{eq:revfill-1}
\tauend \ge \tau_{I(\tauend) - I_0 - 1} = \tau_{i - I_0} { = \tau_{i - I_0 - 1} + \Cslack\sum_{k=1}^i \lceil W_k \rceil}.
\end{equation}
In particular,~\eqref{eq:revfill-1} implies that $t_0+\tauend-\tau+1 \ge t_0+1$ for all $\tau \le 1+\Cslack\sum_{k=1}^i\lceil W_k\rceil$. It follows from the definitions of $\unstickProb(\cdot)$ and $\bins(\cdot)$
that 
$$\unstickProb(t_0+\tauend - \tau+1) = \exp(- \zeta |\bins(\tauend - \tau+1)|/16)
= \exp(- \zeta |B_{ I(\tauend - \tau+1)-1} |/16)
.$$
 Since  $|B_{i'}|$ is monotonically increasing in~$i'$ (which follows from $\Cblock >2$) 
 and $I(\cdot)$ is monotonically increasing in its argument, the right hand side is maximised
 when $\tau$ is as large as possible, so
 \begin{equation}\label{eq:calc}
 \unstickProb(t_0+\tauend - \tau+1) \leq  \exp(- \zeta |B_{ I(\tauend 
  - {\Cslack} {  \sum_{k=1}^i \lceil W_{k}\rceil})-1} |/16)
.\end{equation}

To simplify the notation, let 
$x = \tauend 
  - {\Cslack}{\sum_{k=1}^i \lceil W_{k} \rceil }$.
  From the definition of $I(\cdot)$  we have $x < \tau_{I(x)- I_0}$, 
  so since $x\geq \tau_{(i-1)-I_0}$ by~\eqref{eq:revfill-1}
  we have $i-1 < I(x)$ and hence $i \leq I(x)$.
  Plugging in the value of~$x$, we conclude that  
 $I(\tauend 
  - {\Cslack}{\sum_{k=1}^i \lceil W_{k} \rceil})   \ge i$. 
   Plugging this into~\eqref{eq:calc}, 
   we get
   \begin{equation*} 
 \unstickProb(t_0+\tauend - \tau+1) \leq  \exp(- \zeta 
 |B_{i-{1}}| /16).\end{equation*}
 Since  $i\geq I_0 \geq 4$,
 Observation~\ref{obs:block} gives
 $|B_{i-{1}}| \geq u(i-{1})  (\Cblock-1)/\Cblock$,
  and 
  we can simplify this as follows:
  \begin{equation}\label{blah} 
 \unstickProb(t_0+\tauend - \tau+1) \leq  \frac{1}{e^{ \zeta 
  u(i-{1}) \frac{\Cblock-1}{16\Cblock}}}
  = 
  \frac{1}{e^{ \zeta 
  u(i) \frac{\Cblock-1}{16\Cblock^{{2}}}}}
  .\end{equation}
Since $u(i) \ge i \ge I_0$, the definition of $I_0$ ensures that $\exp(\zeta u(i)\tfrac{\Cslack-1}{16\Cslack^2}) \ge 4u(i)$; we therefore have $\unstickProb(t_0+\tauend - \tau+1) \le 1/(4u(i)) \le 1/(4j)$, as required for the union bound.

For (II) and (III) 
we first note that (II) implies (III) by~\eqref{eq:revfill-1}.
So to finish the proof, we will prove the probability that (II) fails for a given ball $\beta \in b_j^{\revRandom}(0)$ is at most~$1/2$. 
The number of time steps that it takes for $\beta$ to send from bin~$1$ is 
a random variable~$\Psi$ which is the sum 
of independent geometric random variables with parameters~$p_j,\ldots,p_1$.
The expectation of~$\Psi$ is $\mu :=\sum_{k=1}^j (1/p_k)$, and
we just need to show that $\pr(\Psi > 1 + {\Cslack}{\sum_{k=1}^i \lceil W_{k} \rceil} )\leq 1/2$.

Observe that since $j \in B_i$, we have $\mu \le \sum_{k=1}^i W_k \le \sum_{k=1}^i \lceil W_k \rceil$. Hence 
we have 
\[
\pr\Big(\Psi > 1 + {\Cslack}\sum_{k=1}^i \lceil W_{k}\rceil\Big) \leq
\pr(\Psi> \Cslack\mu).
\]
By Markov's inequality, this is at most $1/\Cslack \leq 1/2$.

\end{proof}

\subsection {Bounding \texorpdfstring{$\Fill_j^{\Random}(\tauend)$}{Fill in a random-unsticking process} } 

We next use the time-reversal bijection from Section~\ref{sec:coupling} to translate Lemma~\ref{lem:revfill} from a reverse random-unsticking process $\revRandom$ to a random-unsticking process~$\Random$ by coupling the two processes.

 \begin{lemma}\label{lem:unconditionedfill}
Fix $\lambda$, $\eta$ and $\nu$ in $(0,1)$  and a $(\lambda,\eta,\nu)$-suitable send sequence~$\mathbf{p}$.
Let $Y$ be an externally-jammed process with arrival rate~$\lambda$ and send sequence $\mathbf{p}$.
Consider any integer~$t_0 \geq \jmin$. 
Let~$\tauend\geq \tauinit$ be an   integer.
Then  there is a coupling of
\begin{enumerate}[(i)]
\item a random-unsticking process $\Random=\Random(Y,t_0)$   with
\item a sample $\{Z_j^- \mid j \in \bins(\tauend)\}$ where each $Z_j^-$ is chosen independently from a Poisson distribution with mean $\lambda/(4 p_j)$ 
\end{enumerate}
in such a way that, for every integer $j\in B_{I(\tauend)-1}$,  
  $|\Fill_j^{\Random}(\tauend)| \ge Z_j^-$.
\end{lemma}
 
\begin{proof}
Fix $t_0$ and $\tauend$. 
Let $Y'$ be an externally-jammed process with arrival rate~$\lambda$ and send sequence $\mathbf{p}$,
and let $\revRandom = \revRandom(Y',t_0,\tauend)$ be a reverse random-unsticking process.
Consider any $t\geq t_0+1$ 
and positive integer~$J$.
Let $\pi$ be the  time-reversal bijection 
from 
$\TrajectorySet^{\Random}(t_0,\tauend,t,J)$ to $\TrajectorySet^{\revRandom}(t_0,\tauend,
\tauend - t+t_0+1,J)$, and recall 
from Lemma~\ref{lem:timerev} that for all trajectories $\trajectory \in \TrajectorySet^\Random(t_0,\tauend,t,J)$, we have $\mu^{\Random}(\trajectory) = \mu^{\revRandom}(\pi(\trajectory))$
and
  $\revS(\pi(\trajectory)) = \{\tauend - t + t_0+1 \mid t \in S(\trajectory)\}$.
From the first of these statements, the number of balls following a trajectory $\trajectory \in  
\TrajectorySet^{\Random}(t_0,\tauend,t,J)$ in the  process~$\Random$ has the same distribution as the number of balls following the trajectory 
$\pi(\trajectory)$ in the process~$\revRandom$. We may therefore couple $\Random$ with $\revRandom$ by identifying these trajectories, so that for every ball following trajectory $B$ in $\Random$ there is a corresponding ball following trajectory $\pi(B)$ in $\revRandom$ and vice versa.

We next show that a ball following a trajectory $\trajectory \in  
\TrajectorySet^{\Random}(t_0,\tauend,t,J)$ contributes to $\Fill_j^{\Random}(\tauend)$
if and only if a ball following the trajectory $\pi(\trajectory)$ contributes to 
$\Fill_j^{\revRandom}(\tauend)$, so that $|\Fill_j^{\Random}(\tauend)| = |\Fill_j^{\revRandom}(\tauend)|$; given this, the lemma will follow immediately from Lemma~\ref{lem:revfill}.
By the definition of~$\pi$ (Definition~\ref{def:time-rev-bi}), $J(\pi(\trajectory)) = J(\trajectory)$, 
so (i) is the same in the definitions of $\Fill_j^{\Random}(\tauend)$
and $\Fill_j^{\revRandom}(\tauend)$.
Similarly, (ii) and (iii) are the same in both definitions.
To see that (iv) is the same in both directions,
 consider $t' \in S(\trajectory)$ with $\unstickIndicator(\trajectory,t')=0$.
Then
$\tauend - t' + t_0+1 \in \revS(\revT)$ by Lemma~\ref{lem:timerev}
and,  by the definition of~$\pi$,
$\unstickIndicator(\pi(\trajectory),\tauend - t'+t_0+1) = \unstickIndicator(\trajectory,t')$.
Clearly, the same works in the other direction, so (iv) is the same in both definitions. Hence we do indeed have $|\Fill_j^{\Random}(\tauend)| = |\Fill_j^{\revRandom}(\tauend)|$, and we are done by Lemma~\ref{lem:revfill}.
\end{proof} 

 Finally, we show that Lemma~\ref{lem:unconditionedfill} remains true when the value of $\calF_{t_0}^R$ is exposed. (Since $\bins(\tauend) = B_{I(\tauend)-1}$ by definition, this is the only difference between the statement of Lemma~\ref{lem:unconditionedfill} and that of the following lemma.)
 
\begin{lemma} 
\label{lem:dombinomsZ}
Fix $\lambda$, $\eta$ and $\nu$ in $(0,1)$  and a $(\lambda,\eta,\nu)$-suitable send sequence~$\mathbf{p}$. 
Let $Y$ be an externally-jammed process with arrival rate~$\lambda$ and send sequence $\mathbf{p}$.
Consider any integer $t_0\geq \jmin$. 
Consider a random-unsticking process $\Random = \Random(Y,t_0)$.
Let $F^{\Random}_{t_0}$ be any value of the filtration $\calF^{\Random}_{t_0}$ covering the first~$t_0$ steps of  $\Random$. 
Let~$\tauend\geq \tauinit$ be an   integer.
Then  there is a coupling of
\begin{enumerate}[(i)]
\item The process $\Random$, conditioned on $\calF^{\Random}_{t_0}=F^{\Random}_{t_0}$, with
\item A sample $\{Z_j^- \mid j \in \bins(\tauend)\}$ where each $Z_j^-$ is chosen independently from a Poisson distribution with mean $\lambda/(4 p_j)$ 
\end{enumerate}
in such a way that, for every integer $j\in \bins(\tauend)$,  
  $|\Fill_j^{\Random}(\tauend)| \ge Z_j^-$.
\end{lemma}

\begin{proof}
 Let $i = I(\tauend)-1$.
 For $j \in  B_{i}$, from
 its definition in Section~\ref{sec:fill-def},
   $\Fill_j^{\Random}(\tauend)$
   is the subset of balls   $\beta \in \stuck_j^{\Random}(t_0+\tauend)$
such that
$\birth^{\Random}(\beta) \geq  
 \max\{t_0+1, t_0+\tauend - \Cslack\sum_{k=1}^i \lceil W_k \rceil\}$.
The important thing here is that 
$\birth^{\Random}(\beta) \geq t_0+1$.

The definition of the   process~$\Random$
ensures that balls behave independently. 
Thus whether or not a ball $\beta$ contributes to $\Fill_j^{\Random}(\tauend)$
is independent of $\calF_{t_0}^{\Random}$.
Thus, writing $\Random'$ for a copy of $\Random$ conditioned on 
$\calF_{t_0}^{\Random} = F_{t_0}^{\Random}$, we can couple $\Random$ and $\Random'$
in such a way that $\Fill_j^{\Random}(\tauend) = \Fill_j^{\Random'}(\tauend)$.
We conclude that Lemma~\ref{lem:dombinomsZ} follows
from Lemma~\ref{lem:unconditionedfill}. 
\end{proof}

\subsection{Proving Lemma~\ref{lem:newdombinoms}}

\begin{lemma}\label{lem:domcouplenew}
Fix $\lambda$, $\eta$ and $\nu$ in $(0,1)$  and a $(\lambda,\eta,\nu)$-suitable send sequence~$\mathbf{p}$ with $p_0=1$. 
Let $Y^A$ and $Y^B$ be independent externally-jammed processes with arrival rate $\lambda/2$ and send sequence $\mathbf{p}$
and consider the two-stream externally-jammed process $\YAB=\YAB(Y^A,Y^B)$.
Consider any integer~$t_0 \geq  \jmin$.
Let $F^{\YAB}_{t_0}$ be a value of $\calF^{\YAB}_{t_0}$ 
such that  $\Einit^{\YAB}(t_0)$ occurs.
Let $C$ and $C'$ be distinct elements of $\{A,B\}$.
Let $\Random=\Random(Y^C,t_0)$ be a random-unsticking process.
Let $F^{\Random}_{t_0}$ be a value of $\calF^{\Random}_{t_0}$
that is consistent with $F^{\YAB}_{t_0}$.
Consider
any integer $    \tau \geq 1$.
Then there is a coupling of   
\begin{itemize}
    \item  $\YAB$ conditioned on $\calF^{\YAB}_{t_0} = F^{\YAB}_{t_0}$ and
    \item $\Random$ conditioned on $\calF^{\Random}_{t_0} = F^{\Random}_{t_0}$
\end{itemize}
     such that  
 at least one of the following occurs:
\begin{itemize}
\item  
$\Ejam^{\YAB}(C',t_0,\tau)$  does not occur, or
\item $\unstuck^{\YAB}(C, t_0+\tau) \subseteq \unstuck^{\Random}( t_0+\tau)   $. \end{itemize}
\end{lemma}

\begin{proof}
Note (from the definition of $\calF_{t_0}^R$
and $\calF_{t_0}^T$) since $F_{t_0}^R$ is consistent with $F_{t_0}^T$, $\unstuck^\Random(t_0) = \unstuck^{\YAB}(C,t_0)$ and,
for all $j$, $b_j^\Random(t_0) = b_j^{Y^C}(t_0)$.

We will construct  the coupling step-by-step.
For every positive integer~$\tau'$,
we use the following notation.
 \begin{itemize}
\item Invariant~1 for~$\tau'$: $\Ejam^{\YAB}(C',t_0,\tau')$ occurs, and
\item Invariant~2 for~$\tau'$: $\unstuck^{\YAB}(C,t_0+\tau'-1) \subseteq \unstuck^{\Random}(t_0+\tau'-1)   $.
 \end{itemize}

 Note that the invariants for~$\tau'$ only 
 depend on steps $1,\ldots, t_0+\tau'-1$ of the coupled process.
 Our high-level strategy is as follows:
 Given that the coupling on these steps satisfies both invariants for $\tau'$
 we will  show how to extend the coupling to step~$t_0+\tau'$ 
 to satisfy Invariant~2 for $\tau'+1$. After that there are two possibilities.
 \begin{itemize}
 \item 
 If $\tau'=\tau$ then we have finished, by establishing 
 $\unstuck^{\YAB}(C,t_0+\tau) \subseteq \unstuck^{\Random}(t_0+\tau)   $.
 \item Otherwise, $\tau' \leq \tau-1$ and there are again two possibilities. 
 \begin{itemize}
     \item  Either
 Invariant~1 is violated for $\tau'+1$
 in which case 
 $\Ejam^{\YAB}(C',t_0,\tau'+1)$ does not occur so
 $\Ejam^{\YAB}(C',t_0,\tau)$ does not occur (and we have finished),
 \item
 or Invariant~1 is satisfied for $\tau'+1$ and the construction of the coupling continues to the next step.
 \end{itemize}
 \end{itemize}
That completes the description of the high-level strategy.
To finish the proof we must show that both invariants hold for $\tau'=1$
and then we must show how to extend the coupling as promised.

We first establish the invariants for $\tau'=1$. 
\begin{itemize}
    \item  Invariant~1 follows from Equation~\eqref{eq:deterministic_tauinit}
    in Lemma~\ref{lem:tau-init}.

\item For Invariant~2 we wish to show 
$\unstuck^{\YAB}(C,t_0) \subseteq \unstuck^{\Random}(t_0)   $.
This follows from the consistency of $F_{t_0}^T$ and $F_{t_0}^R$, as already mentioned.

\end{itemize}

To finish the proof
we consider $\tau'\geq 1$ such that both invariants hold for~$\tau'$, which implies (using the definition of $\Ejam^T(C',t_0,\tau')$ in Definition~\ref{def:jam})
that   
$$f(\stuckvect^{\YAB}(C',t_0+\tau'-1) \geq \zeta |\bins(\tau'-1)|$$  
and
 $$\unstuck^{\YAB}(C,t_0+\tau'-1) \subseteq \unstuck^{\Random}(t_0+\tau'-1)   .$$
Using $\calF^{\YAB}_{t_0+\tau'-1}$ and
$\calF^{\Random}_{t_0+\tau'-1}$ we now wish to consider 
step $t_0+\tau'$, and to extend the coupling to this step in such a way that
$\unstuck^{\YAB}(C,t_0+\tau') \subseteq \unstuck^{\Random}(t_0+\tau')   $.
Once we do this, we'll have completd the high-level strategy, hence we'll have completed the proof.

It is helpful to note by the 
definition of the processes that 
for all $t$ and $j$
we have 
$$\balls^{\YAB}(C,t) \cap b_j^{\YAB}(t) = b_j^{Y^C}(t) = b_j^{\Random}(t).
$$
We will use this observation for $t=t_0+\tau'-1$.
From this we deduce that 
$$\stucksend^{\YAB}(C,t) \subseteq \stucksend^{\Random}(t) \cup \unstuck^{\Random}(t).$$

We are finally ready to do the coupling.
We first determine   ${b'_0}^{Y^C}(t+1)={b'_0}^{\Random}(t+1)$ and
$\send^{Y^C}(t+1)=\send^{\Random}(t+1)$.
This allows us to deduce
$\stucksend^{\YAB}(C,t+1)$ and $\stucksend^{\Random}(t+1)$.

We will use this to finish the coupling of step~$t+1$ in such a way that
any ball $$\beta\in \stucksend^{\YAB}(C,t+1) \cap \unstuck^{\YAB}(C,t+1)$$
is also in $\unstuck^{\Random}(t+1)$.
Here is how we do it: If $\stucksend^{\YAB}(C,t+1)$ is empty, there is nothing to prove.
Otherwise, choose $\beta$ u.a.r. from $\stucksend^{\YAB}(C,t+1)$
(this is the random ball that may become unstuck at time~$t+1$ in~$\YAB$).
The probability that $\beta$ is in 
$\unstuck^{\YAB}(C,t+1)$ is the probability that
$\stucksend^{\YAB}(C',t+1)=\emptyset$.
We will use  Lemma~\ref{lem:one_onesided}
to bound this probability. 
First note that the lemma applies since
the definition of $I_0$ ensures that
$\zeta |B_{I_0}| \geq 4$ so
$\zeta |\bins(\tau')| \geq 4$.
By Lemma~\ref{lem:one_onesided},
the probability that 
$\stucksend^{\YAB}(C',t+1)=\emptyset$
is at most 
$$ \exp(-\zeta |\bins(\tau'-1)|/16).$$
If $\beta$   
is in $\unstuck^{\Random}(t)$ then there is nothing to prove. Alternatively, if
it is in $\stucksend^{\Random}(t)$, then the probability that it is
in $\unstuck^{\Random}(t+1)$ 
is $\unstickProb(t) = \exp(-\zeta |\bins(\tau'-1)|/16)$.
Thus we can extend the coupling as required, completing the proof.

\end{proof}

It is not hard to see that the final bullet point of Lemma~\ref{lem:domcouplenew} implies $\Fill_j^\Random(\tau) \subseteq \Fill_j^\YAB(C,\tau)$; this is the only difference between the statements of Lemma~\ref{lem:domcouplenew} and Corollary~\ref{cor:domcouplefill}.

 \begin{corollary}\label{cor:domcouplefill}
Fix $\lambda$, $\eta$ and $\nu$ in $(0,1)$  and a $(\lambda,\eta,\nu)$-suitable send sequence~$\mathbf{p}$ with $p_0=1$.  
Let $Y^A$ and $Y^B$ be independent externally-jammed processes with arrival rate $\lambda/2$ and send sequence $\mathbf{p}$
and consider the two-stream externally-jammed process $\YAB=\YAB(Y^A,Y^B)$.
Consider any integer~$t_0 \geq  \jmin$.
Let $F^{\YAB}_{t_0}$ be a value of $\calF^{\YAB}_{t_0}$ 
such that  $\Einit^{\YAB}(t_0)$ occurs.
Let $C$ and $C'$ be distinct elements of $\{A,B\}$.
Let $\Random=\Random(Y^C,t_0)$ be a random-unsticking process.
Let $F^{\Random}_{t_0}$ be a value of $\calF^{\Random}_{t_0}$
that is consistent with $F^{\YAB}(t_0)$.
Consider
any integer $    \tau \geq 1$.
There is a coupling of   
\begin{itemize}
    \item  $\YAB$ conditioned on $\calF^{\YAB}_{t_0} = F^{\YAB}_{t_0}$ and
    \item $\Random$ conditioned on $\calF^{\Random}_{t_0} = F^{\Random}_{t_0}$
\end{itemize}
     such that  
 at least one of the following occurs:
\begin{itemize}
\item  
$\Ejam^{\YAB}(C',t_0,\tau)$  does not occur, or
\item For all positive integers~$j$,
$\Fill_j^{\Random}(\tau) \subseteq \Fill_j^{\YAB}(C, \tau)   $. \end{itemize}
\end{corollary}
\begin{proof}
    We use the coupling of Lemma~\ref{lem:domcouplenew}. If $\Ejam(C',t_0,\tau)$ does not occur, then the result is immediate. Otherwise, we have $\unstuck^{\YAB}(C, t_0+\tau) \subseteq \unstuck^{\Random}( t_0+\tau)$. Since every ball is either stuck or unstuck, it follows that $\stuck^R(t_0+\tau) \subseteq \stuck^T(C,t_0+\tau)$. By the definitions of $\Fill_j^{\Random}(\tau)$ and $\Fill_j^{\YAB}(C, \tau)$, it follows that $\Fill_j^{\Random}(\tau) \subseteq \Fill_j^{\YAB}(C, \tau)$ as required.
\end{proof}

Then combining the couplings of Corollary~\ref{cor:domcouplefill} (from one stream of $\YAB$ to $\Random$) and Lemma~\ref{lem:dombinomsZ} (from $\Random$ to i.i.d.\ Poisson variables), we  get the following lemma.

 \begin{lemma}\label{lem:halfdombinoms}
 Fix $\lambda$, $\eta$ and $\nu$ in $(0,1)$  and a $(\lambda,\eta,\nu)$-suitable send sequence~$\mathbf{p}$ with $p_0=1$.
Let $Y^A$ and $Y^B$ be independent externally-jammed processes with arrival rate $\lambda/2$ and send sequence $\mathbf{p}$ 
and consider the two-stream externally-jammed process $\YAB=\YAB(Y^A,Y^B)$.
Consider any integer~$t_0 \geq  \jmin$.
Let $F^{\YAB}_{t_0}$ be a value of $\calF^{\YAB}_{t_0}$ 
such that  $\Einit^{\YAB}(t_0)$ happens.
Let $C$ and $C'$ be distinct elements of $\{A,B\}$.
Consider
any integer $    \tau \geq \tauinit$.
There is a coupling of   
\begin{itemize}
    \item  $\YAB$ conditioned on $\calF^{\YAB}_{t_0} = F^{\YAB}_{t_0}$ and 
  \item A sample $\{Z_j^- \mid j \in \bins(\tau)\}$ where each $Z_j^-$ is chosen independently from a Poisson distribution with mean $\lambda/(4 p_j)$ 
\end{itemize}
     such that  
 at least one of the following occurs:
\begin{itemize}
\item  
$\Ejam^{\YAB}(C',t_0,\tau)$  does not occur, or
\item For all integers $j\in \bins(\tau)$,
$  |\Fill_j^{\YAB}(C, \tau) | \geq Z_j^-   $.
\end{itemize}
\end{lemma}
  
Finally, applying
Lemma~\ref{lem:halfdombinoms} 
twice, once with $C=A$ and once with $C=B$
and using the definitions of $\Fill_j^{\YAB}(C, \tau)$ and $\Ejam^{\YAB}(t_0,\tau)$, we get 
  Lemma~\ref{lem:newdombinoms} as desired.
 
 \lemnewdombinoms*

\section{The main lemma for proving  Theorem~\ref{thm:technical}}\label{sec:provemain}

Our goal in this section is to prove Lemma~\ref{lem:main}, which says that if $\YAB$ is a two-stream externally-jammed process with suitable send sequence, and if $\Einit^\YAB(t_0)$ occurs, then $\YAB$ is likely to stay $t_0$-jammed forever. We start with a simple observation.

 \begin{observation}\label{obs:I0}
 Suppose that $\mathbf{p}$ is $(\lambda,\eta,\nu)$-suitable.
 If $i\geq I_0$ then 
 $|\{ j \in B_i \mid p_j \leq p_* \}| > 2\eta |B_i|/3$.
  \end{observation}
  \begin{proof}
  Since $\mathbf{p}$ is $(\lambda,\eta,\nu)$-suitable,
   $|\{ j \in [u(i)] \mid p_j \leq p_* \}| > \eta u(i)$.
 At most $u(i-1)$ of these bins are in blocks $B_1,\ldots,B_{i-1}$, and by definition we have $u(i-1) = u(i)/\Cblock \leq \eta u(i)/3$, so at least $2 \eta u(i)/3$ of these bins are in block $B_i$. 
  \end{proof}
 
We will prove Lemma~\ref{lem:main} using a union bound over all time steps. For readability, we extract the most difficult case of this bound as Lemma~\ref{lem:toget}.

\begin{restatable}{lemma}{lemtoget}
\label{lem:toget}
 {Let $\mathbf{p}$ be $(\lambda,\eta,\nu)$-suitable.} 
Let $\YAB$ be a  two-stream externally-jammed process  with arrival rate $\lambda \in (0,1)$ and send sequence~$\mathbf{p}$.
Consider any integer~$t_0 \geq  \jmin$.
Let $F^{\YAB}_{t_0}$ be a value of $\calF^{\YAB}_{t_0}$ 
such that  $\Einit^{\YAB}(t_0)$ happens.
Let $\tau\geq \tauinit$.
Conditioned on $\calF^{\YAB}_{t_0}=F^{\YAB}_{t_0}$, with probability at least $1 - 1/(10\tau^2)$, 
at least one of the following happens:
\begin{itemize}
    \item The event $\Ejam^T(t_0,\tau) $ doesn't occur.  
    \item The process~$\YAB$ is $t_0$-jammed for time $\tau+1$.
    \end{itemize}
\end{restatable}

\begin{proof}
 Let $Y^A$ and $Y^B$ be independent externally-jammed processes with arrival rate $\lambda/2$ and send sequence $\mathbf{p}$
 such that $\YAB=\YAB(Y^A,Y^B)$.
 Consider $C \in \{A,B\}$.
Let $\{Z_j \mid j \in \bins(\tau)\}$ be independent Poisson variables with $\E(Z_j) = \lambda/(4p_j)$. 
Since  {$\zeta= \eta \lambda/24$,}
by  Lemma~\ref{lem:newdombinoms} and the definition of $f$ (Definition~\ref{def:noise}), 
\begin{align}\nonumber
        &\pr\bigg(\overline{\Ejam^T(t_0,\tau)}\mbox{ or }f(\stuckvect^T(C,t_0+\tau)) \ge \zeta |\bins(\tau)|\,\Big|\, \calF_{t_0}^T = F_{t_0}^T\bigg)\\\label{eq:lemtoget-dombinoms-0}
        &\qquad\qquad\qquad\qquad\qquad\qquad\qquad\qquad\qquad\qquad \ge \pr\bigg(\sum_{j \in \bins(\tau)} p_j Z_j \ge \frac{ {\eta}\lambda|\bins(\tau)|}{ {24}}\bigg).
\end{align}
Since $\tau \ge \tauinit \ge 0$, we have $\bins(\tau) = B_{I(\tau)-1}$
where $I(\tau)\geq I_0+1$.
  Observation~\ref{obs:I0} ensures that 
there is a set $\bins'(\tau) \subseteq \bins(\tau)$ 
such that
$|\bins'(\tau)| \geq 2 \eta |\bins(\tau)|/3$ and for
all $j\in \bins'(\tau)$, 
we have $p_j \leq p_* $.  
 From Equation~\eqref{eq:lemtoget-dombinoms-0}
and the lower bound on $|\bins'(\tau)|$ we get
\begin{align}\nonumber
        &\pr\bigg(\overline{\Ejam^T(t_0,\tau)}\mbox{ or }f(\stuckvect^T(C,t_0+\tau)) \ge \zeta |\bins(\tau)|\,\Big|\, \calF_{t_0}^T = F_{t_0}^T\bigg)\\\label{eq:NY}
        &\qquad\qquad\qquad\qquad\qquad\qquad\qquad\qquad\qquad\qquad \ge \pr\bigg(\sum_{j \in \bins'(\tau)} p_j Z_j \ge \frac{  \lambda|\bins'(\tau)|}{ {16}}\bigg).
\end{align}

For each $j \in \bins'(\tau)$, let $1_{Z_j}$ be the indicator variable of the event $Z_j \le \lambda/(8p_j)$, so that
\begin{equation} \label{eq:xmas}
        \sum_{j \in \bins'(\tau)}p_j Z_j \ge \frac{\lambda}{8}\Big(|\bins'(\tau)| - \sum_{j \in \bins'(\tau)} 1_{Z_j}\Big).
\end{equation}
We dominate the sum in the right-hand side of~\eqref{eq:xmas} above by a binomial random variable, which we will bound above using the Chernoff bound of Lemma~\ref{lem:chernoff-large-upper}. We first bound the mean $\mu$ of this variable; by the Chernoff bound of Lemma~\ref{lem:chernoff-small-dev} applied with $\delta=1/2$, for all $j \in  {\bins'(\tau)}$ we have 
$$        \pr(1_{Z_j}=1) = \pr\big(Z_j \le \E(Z_j)/2\big) \le e^{-\E(Z_j)/8} 
\le e^{-\E(Z_j)/12}
= e^{- \lambda/(48 p_j)}
\leq e^{- \lambda/(48 p_* )}
.$$ 
Let $x=e^{\lambda/(48p_* )}/2$ and $\mu = e^{-\lambda/(48p_* )}| {\bins'}(\tau)|$.
Since $p_*  \leq \lambda/48$, we have $x>1$, which ensures
the applicability of Lemma~\ref{lem:chernoff-large-upper}. 
From this lemma, we obtain
$$
\pr\bigg(\sum_{j\in\bins'(\tau)}1_{Z_j} \ge \frac{|\bins'(\tau)|}{2}  \bigg) \le 
               e^{-\mu x (\log x -1)} \leq
        \exp\bigg({-}\frac{|\bins'(\tau)|}{2} \Big(\frac{\lambda}{48p_* } -2\Big)\bigg).
 $$
Recall that $|\bins'(\tau)| \ge 2\eta |\bins(\tau)|/3$; hence by Lemma~\ref{lem:I-bound}, it follows that
$$
\pr\bigg(\sum_{j\in\bins'(\tau)}1_{Z_j} \ge \frac{|\bins'(\tau)|}{2}  \bigg) \le
\exp\bigg({-}
\frac{{2 \eta}\log(\tau)}{3 \cdot 4 \Cblock^2 \log(1/\nu)} 
\Big(\frac{\lambda}{48p_* }- 2\Big)\bigg).
$$

The rest of the proof is simple manipulation. Since $p_*  \le \lambda/200$ we have $\lambda/(48p_* ) - 2 \ge \lambda/(100p_* )$.
Since $p_*  \le  (\lambda\eta)/(1800 \Cblock^2\log(1/\nu))$ we have $\lambda/(100 p_* ) \geq (18/\eta) \Cblock^2\log(1/\nu)$ so
    we get
     $$
        \pr\bigg(\sum_{j\in\bins'(\tau)}1_{Z_j} \ge \frac{|\bins'(\tau)|}{2}  \bigg) \le
        \exp({-} 3
         \log(\tau) ) = \frac{1}{\tau^3}.
   $$
 Finally, since $\tau \geq \tauinit \geq 20$, this probability is at most $1/(20 \tau^2)$.   
    It follows from~\eqref{eq:xmas} that
    \[
         \pr\bigg(\sum_{j \in \bins'(\tau)} p_j Z_j \ge \frac{\lambda|\bins'(\tau)|}{16}\bigg) \ge 1 - \frac{1}{20\tau^2}.
    \]
    It follows from~\eqref{eq:NY} that
    \[
        \pr\bigg(\overline{\Ejam^T(t_0,\tau)}\mbox{ or }f(\stuckvect^T(C,t_0+\tau)) \ge \zeta |\bins(\tau)|\,\Big|\, \calF_{t_0}^T = F_{t_0}^T\bigg) \ge 1 - \frac{1}{20\tau^2}.
    \]
    The result therefore follows by a union bound over $C \in \{A,B\}$.
\end{proof}

\begin{restatable}{lemma}{lemmain} 
\label{lem:main} 
 {Let ${\mathbf{p}}$ be $(\lambda,\eta,\nu)$-suitable.}
Let $\YAB$ be a two-stream externally-jammed process with 
arrival rate~$\lambda\in(0,1)$ and send sequence~$\mathbf{p}$.
Consider any integer~$t_0  \geq \jmin $. 
Let $F^{\YAB}_{t_0}$ be a value of $\calF^{\YAB}_{t_0}$ such 
that $\Einit^{\YAB}(t_0)$ happens. Then  
\[
    \pr\Big(\YAB\mbox{ is $t_0$-jammed for all }\tau \geq 1 \mid \calF^{\YAB}_{t_0}=F^{\YAB}_{t_0} \Big) \ge 4/5.
\]
\end{restatable}

\begin{proof}
    For all $\tau \ge \tauinit$, let $\calE_\tau$ be the event that either 
   $\Ejam^T(t_0,\tau) $ doesn't occur   
or that $\YAB$ is $t_0$-jammed for $\tau+1$. We then have
    \begin{align*}
        \pr\Big(\bigwedge_{\tau \ge 1}\Ejam^T(t_0,\tau) \mid \calF_{t_0}^\YAB = F_{t_0}^\YAB\Big) 
        &= \pr\Big(\Ejam^T(t_0,\tauinit) \wedge \bigwedge_{\tau \ge \tauinit}\calE_\tau \mid \calF_{t_0}^\YAB = F_{t_0}^\YAB \Big).
    \end{align*}
    By Lemmas~\ref{lem:tau-init} and~\ref{lem:toget} and a union bound, we therefore have
    \[
        \pr\Big(\bigwedge_{\tau \ge 1}\Ejam^T(t_0,\tau) \mid \calF_{t_0}^\YAB = F_{t_0}^\YAB\Big)  \ge   \frac{49}{50} - \sum_{\tau \ge \tauinit}\frac{1}{10\tau^2} \ge \frac{49}{50} - \frac{\pi^2}{60} \ge \frac{4}{5},
    \]
    as required.
\end{proof}

\section{Proof of Main Theorems} \label{sec:mainproofs}

\thmtechnical*

\begin{proof}
For every positive integer~$i$, 
let $R_i$ be the $i$'th 
smallest time step in $\mathcal{R}^T$.
If $|\mathcal{R}^T|<i$,
then we define $R_i= \infty$.  Since $R_1 = 0$, for all $i \ge 2$ we have
\begin{equation}\label{eq:J:thmmono-1}
        \pr(R_i < \infty) = \prod_{k=2}^i \pr(R_i < \infty \mid R_{i-1} < \infty).
\end{equation}
Let $\calE_i$ be the event that $R_{i-1} < \infty$ and that $\Einit^T(R_{i-1}+t)$ occurs for some positive integer~$t$ such that $R_{i-1}+t < R_i$. 
Observe that conditioned on $R_{i-1} < \infty$, the probability of $\calE_i$ occurring does not depend on $i$ since $\stuck^T$ evolves independently of unstuck balls; let $p_{\mathrm{init}} = \pr(\calE_i \mid R_{i-1} < \infty)$. Then by~\eqref{eq:J:thmmono-1}, we have
\begin{equation}\label{eq:J:thmmono-2}
\pr(R_i < \infty) \le \prod_{k=2}^i \Big(1 - p_{\mathrm{init}} + p_{\mathrm{init}} \pr(R_i < \infty \mid \calE_i, R_{i-1} < \infty)\Big).
    \end{equation}

Observe that if $\Einit(R_{i-1}+t)$ occurs for some $t < R_i-R_{i-1}$, and $T$ is $(R_{i-1}+t)$-jammed for all $\tau \ge 1$, then   we have we must have $\stuck^T(R_{i-1}+u)\neq \emptyset$ for all $u \ge 1$; hence $R_i = \infty$.

It follows by applying Lemma~\ref{lem:main} to~\eqref{eq:J:thmmono-2} that

    \begin{equation}\label{eq:J:thmmono-3}
        \pr(R_i < \infty) \le \prod_{k=2}^i \Big(1 - p_{\mathrm{init}} + p_{\mathrm{init}}/5 \Big) \le e^{-4(i-1)p_{\mathrm{init}}/5}.
    \end{equation}

By~\eqref{eq:J:thmmono-3}, we have $\sum_i \pr(R_i < \infty) < \infty$. It follows by the Borel-Cantelli lemma that with probability 1 we have $R_i = \infty$ for all but finitely many values of $i$, which proves the theorem.
\end{proof}

 \corlast*
 
 \begin{proof}
Let $T$ be a two-stream externally-jammed process
with arrival rate $\lambda$ and send sequence
$\mathbf{p}$. 
Let $\mathcal{R}^{T} = \{t \mid \stuck^{T}(t)=\emptyset\}$. Theorem~\ref{thm:technical} shows that, with probability~$1$, $\mathcal{R}^{T}$ is finite.
Let $Y$ be an externally-jammed process with arrival rate~$\lambda$ and send sequence~$\mathbf{p}$, and 
let $\mathcal{R}^{Y} = \{t \mid \stuck^{Y}(t)=\emptyset\}$.
By the coupling of $Y$ and $T$ from Lemma~\ref{lem:coupleY},
with probability~$1$, $\mathcal{R}^{Y}$ is finite.
Finally, by the coupling of $X$ and $Y$ from Observation~\ref{obs:ext-jammed-useful}, it follows that
with probability~$1$, $\mathcal{R}^{X}$ is finite.
 \end{proof}

In order to make the proof of Theorem~\ref{thm:LCED} more
convenient, we first re-state the definition of LCED.
 
 \defLCED*
 
\thmLCED*

\begin{proof}
Suppose that $\mathbf{p}$ is not LCED. 
Fix $\lambda \in (0,1)$.
Let $X$ be a backoff process with arrival rate~$\lambda$ 
and send sequence~$\mathbf{p}$. We split the analysis 
into cases depending on which case of Definition~\ref{def:LCED} fails.  
    
\medskip\noindent
\textbf{Case 1 (\ref{LCED:notsuperexponential}) fails:}  
By definition, since (\ref{LCED:notsuperexponential}) fails,
there is no real number $C>0$ such that, for all but finitely many~$j$, $\log(1/p_j) \leq C j$.
So for all $C>0$, there are infinitely many~$j$ such that
$\log(1/p_j) > C j$, which implies $p_j<  \exp(-C j)$.
Taking $C = \log(2/(\lambda p_0))$, we conclude that   there are infinitely many~$j$ such that $p_j < (\lambda p_0/2)^j$.
Lemma~\ref{lem:killer} implies that $X$  is unstable.

\medskip\noindent
\textbf{Case 2 (\ref{LCED:notsuperexponential}) holds but  (\ref{LCED:largelyconstant}) fails:} By definition, since 
(\ref{LCED:largelyconstant}) fails, there is an $\eta > 0$ such that for all $c>0$, for all but finitely many $n$,
$|\{j \le n \colon p_j  \leq c\}| >  \eta n$.
Taking $c= p_*(\lambda p_0,\eta,\nu)$
for positive $\nu$,
we conclude that there is an $\eta>0$
such that for all $\nu>0$ 
  and for all but finitely many $n$,
$|\{j \le n \colon p_j  \leq p_*(\lambda p_0,\eta,\nu)\}| >  \eta n$. 

Let $\mathbf{p}'$ be the send sequence
derived from~$\mathbf{p}$ by setting $p'_0=1$
and, for every positive integer~$j$, $p'_j=p_j$.

Given the $\eta>0$ that we have
just identified, our goal will be find $\nu>0$
such that $\mathbf{p}'$ is $(\lambda p_0, \eta,\nu)$-suitable
(Definition~\ref{def:suitable}).
To achieve this, we need to show that we can choose $\nu>0$
such that, for all but finitely many~$n$, $\nu^n < p_n=p'_n$.
 
Here we will use the fact that  (\ref{LCED:notsuperexponential}) holds.
 Since $\log(1/p_j) = O(j)$, there exists 
    $C>0$ such that for all but finitely many $j$,
    $1/p_j < e^{C j}$. So taking $\nu = e^{-C} <1$
    we have that   for all but finitely many $j$, $p_j > \nu^j$.
 So we have now shown that   $\mathbf{p}'$ is $(\lambda p_0, \eta,\nu)$-suitable.  
 
 Let $X'$ be a backoff process with arrival rate $\lambda p_0$
and send sequence $\mathbf{p}'$. 
Let $\mathcal{R}^{X'} = \{t \mid \balls^{X'}(t)=\emptyset\}$.
By Corollary~\ref{cor:last}, 
with probability~$1$, $\mathcal{R}^{X'}$ is finite.

Finally, let $X$ be a backoff process with arrival rate $\lambda$
and send sequence $\mathbf{p}$. 
Let $\mathcal{R}^{X} = \{t \mid \balls^{X}(t)=\emptyset\}$.
By the coupling of $X$ and $X'$ from Observation~\ref{obs:p0one}, 
with probability~$1$, $\mathcal{R}^{X}$ is finite.
Hence, the process~$X$ is transient, so it is unstable.

\medskip\noindent
\textbf{Case 3 
(\ref{LCED:notsuperexponential}) holds but   (\ref{LCED:exponentialdecay}) fails:} 
In this case  
$\log(1/p_j) = o(j)$ so 
the result follows from Corollary~\ref{cor:backoff-kmp}.
Here is an argument that $\log(1/p_j)=o(j)$.
Since  (\ref{LCED:notsuperexponential}) holds but   (\ref{LCED:exponentialdecay}) fails,
there is no  infinite subsequence $(p_{\ell_1},p_{\ell_2},\dots)$    that
satisfies $\log (1/p_{\ell_x}) = \Omega(\ell_x)$.        
Thus for every infinite subsequence 
$(p_{\ell_1},p_{\ell_2},\dots)$ and for every $C>0$, 
 infinitely many $n$ satisfy $\log(1/p_{\ell_n}) < C \ell_n$.
This implies that for all $C>0$ and all infinite subsequences $(p_{\ell_1},p_{\ell_2},\dots)$,
 some $n$ has $\log(1/p_{\ell_n}) < C \ell_n$.
 Hence, for all $C>0$ there is no infinite subsequence
 $(p_{\ell_1},p_{\ell_2},\dots)$    that
 satisfies $\log (1/p_{\ell_x}) \ge C\ell_x$.
 Finally, for all $C>0$, there is a $j_C$
 such that for all $j\geq j_C$,   $\log (1/p_j) < C j$ so $\log(1/p_j) = o(j)$, as required.
\end{proof}   
     
 We can now use Theorem~\ref{thm:LCED} to derive the other theorems stated in the introduction.

\thmmono*
\begin{proof}
We prove that $\mathbf{p}$ is not LCED and use Theorem~\ref{thm:LCED}.
Suppose for contradiction that $\mathbf{p}$  
is LCED.
Let $c>0$ be such that $p_j \ge c$ for infinitely many values of $j$; by (\ref{LCED:largelyconstant}) of Definition~\ref{def:LCED} applied with any $\eta>0$, such a $c$ must exist. Let 
$(p_{\ell_1},p_{\ell_2},\dots)$
be 
the infinite exponentially-decaying subsequence of $\mathbf{p}$ guaranteed by  (\ref{LCED:exponentialdecay}), and observe that there exists $i_c$ such that $p_{\ell_i} < c$ for all $i \ge i_c$. Since $\mathbf{p}$ is monotonically non-increasing, 
it follows that for all $j \ge i_c$ we have $p_j \le p_{i_c} \le c$, giving a contradiction.
\end{proof}

\thmsecond*
\begin{proof}
We observe that $\mathbf{p}$ is not LCED and use Theorem~\ref{thm:LCED}.
In particular, we will observe that item (\ref{LCED:largelyconstant}) in the definition of LCED cannot be satisfied for $\eta=1/2$.
Item (\ref{LCED:largelyconstant}) implies that there is a
$c>0$ such that for infinitely many $n$, $m_{\mathbf{p}}(n) > c$.
But this is inconsistent with $m_{\mathbf{p}}(n) = o(1)$.

\end{proof}

\section{Conclusion}

The main result of this paper is Theorem~\ref{thm:LCED} which shows that for every send sequence $\mathbf{p}$ which is not LCED and every arrival rate $\lambda \in (0,1)$, the backoff process with arrival rate~$\lambda$ and send sequence~$\mathbf{p}$ is unstable.
This goes a long way towards proving Aldous's conjecture (Conjecture~\ref{conj:backoff}), that for every send sequence~$\mathbf{p}$ and every arrival rate $\lambda \in (0,1)$, the backoff process with arrival rate~$\lambda$ and send sequence~$\mathbf{p}$ is unstable. For example, Theorem~\ref{thm:LCED} implies instability if $\mathbf{p}$ is monotonically non-decreasing (Theorem~\ref{thm:mono})
or if the median of the first $n$  entries in the send sequence is $o(1)$ (Theorem~\ref{thm:second}). Nevertheless significant obstacles remain for proving the full conjecture. This is because LCED send sequences can exhibit qualitatively different behaviour from the send sequences covered by Theorem~\ref{thm:LCED}, as explained in Section~\ref{sec:intro-future}.

\bibliographystyle{plainurl}
\bibliography{contention}

 \end{document}